\documentclass[10pt,conference]{IEEEtran}
\IEEEoverridecommandlockouts
\usepackage{cite}
\usepackage{amsmath,amsthm,amssymb,amsfonts}
\usepackage{graphicx}
\usepackage{txfonts}
\usepackage{subcaption}
\usepackage{xcolor}
\usepackage{siunitx}
\usepackage{neuralnetwork}
\newtheorem{definition}{Definition}
\newtheorem{theorem}{Theorem}

\usepackage[mathlines]{lineno}
\newcounter{example}
\newenvironment{example}[1][]{\refstepcounter{example}\par\medskip
	\noindent \textbf{Example~\theexample. #1} \rmfamily}{\medskip}
\usepackage{makecell}
\usepackage{multirow}
\usepackage{enumerate}
\usepackage{tabularx}
\usepackage{hhline}

\usepackage{subcaption}
\usepackage{url}
\usepackage{colortbl}
\usepackage{xcolor}
\usepackage{multicol}
\definecolor{atomictangerine}{rgb}{1.0, 0.6, 0.4}
\definecolor{apricot}{rgb}{0.98, 0.81, 0.69}
\definecolor{antiquewhite}{rgb}{0.98, 0.92, 0.84}
\definecolor{Gray}{gray}{0.90}
\definecolor{LightCyan}{rgb}{0.88,1,1}
\newcolumntype{e}{>{\columncolor{Gray}}r}
\usepackage{threeparttable}
\usepackage[ruled,vlined,linesnumbered]{algorithm2e}
\usepackage{pgfplots}

\SetCommentSty{mycommfont}

\newtheorem{proposition}{Proposition}

\title{\resizebox{\textwidth+10pt}{!}{
\begin{minipage}{\textwidth+20pt}
\textsf{DualApp:} Tight Over-Approximation for Neural Network Robustness Verification via Under-Approximation
\end{minipage}
}} 

\makeatletter
\newcommand{\linebreakand}{%
  \end{@IEEEauthorhalign}
  \hfill\mbox{}\par
  \mbox{}\hfill\begin{@IEEEauthorhalign}
}
\makeatother

\author{
    Yiting Wu \\
    51205902026@stu.ecnu.edu.cn \\
    East China Normal University \\
    Shanghai, China
    \and
    Zhaodi Zhang \\
    zdzhang@stu.ecnu.edu.cn \\
    East China Normal University \\
    Shanghai, China
    \and
    Zhiyi Xue \\
    51255902046@stu.ecnu.edu.cn \\
    East China Normal University \\
    Shanghai, China\\
    \linebreakand
    Si Liu\\
    {si.liu@inf.ethz.ch}\\
    ETH Z{\"u}rich\\
    \and
    Min Zhang \\
    zhangmin@sei.ecnu.edu.cn \\
    East China Normal University \\
    Shanghai, China
}

\begin{document}

	\maketitle

\begin{abstract}
	The robustness of neural networks is fundamental  to the hosting system's reliability and security. Formal verification has been proven to be effective in providing  provable robustness guarantees. To improve the verification scalability, over-approximating the non-linear activation functions in neural networks by linear constraints is widely adopted, which transforms the verification problem into an efficiently solvable linear programming problem. As over-approximations inevitably introduce overestimation, many efforts have been dedicated to defining 
	the tightest possible approximations. 
	Recent studies have however showed that the existing so-called \emph{tightest} approximations are superior to each other. 
	
	In this paper we identify and report an crucial factor in defining tight approximations, namely the \emph{ approximation domains of activation functions}. We observe that  existing approaches only rely on overestimated domains,  while the corresponding tight approximation may not necessarily be tight on its actual  domain. 
	We propose a novel under-approximation-guided approach, called \emph{dual-approximation}, to define tight over-approximations and  two complementary under-approximation algorithms based on sampling and gradient descent. The overestimated domain guarantees the soundness while the underestimated one guides the tightness. We implement our approach into a tool called \textsf{DualApp} and extensively evaluate it on a comprehensive benchmark of 84 collected and trained neural networks with different architectures. 
	The experimental results show that \textsf{DualApp} outperforms the state-of-the-art approximation-based approaches, with up to 71.22\% improvement to the verification result. 
	
	
	
\end{abstract}

	\section{Introduction}\label{sec:intro}
Deep neural networks (DNNs) are the most crucial components in AI-empowered software systems. They must be guaranteed reliable and dependable when the hosting systems are safety-critical. 
 Robustness is central to their safety and reliability,
  ensuring that neural networks can function correctly even under environmental perturbations and adversarial attacks \cite{szegedy2013intriguing,goodfellow2014explaining,wu2020robustness}. 
Studying the  robustness of DNNs from both training and engineering perspectives attracts researchers from both AI  and SE communities \cite{szegedy2013intriguing,goodfellow2014explaining,ilyas2019adversarial,LiYHS0Z22,LiuFY022,Pan20}. More recently, the emerging  formal verification efforts on the robustness of neural networks aim at providing certifiable robustness guarantees for the neural networks (see the surveys \cite{huang2020survey,liu2021algorithms,wing2021trustworthy} for details). Certified robustness of neural networks is necessity for guaranteeing that the hosting software systems are both safe and secure. A provably robust neural network ensures that no adversarial examples can falsify a network to mis-classify an input and thereafter cause unexpected behaviors in the hosting system. 
Robustness verification is particularly crucial to the neural networks planted in safety-critical applications such as autonomous drivings \cite{bojarski2016end,apollo}, medical diagnoses \cite{titano2018automated}, and access controls by face recognition \cite{sun2015deepid3}. 

Formally verifying the robustness  of neural networks is computationally complex and expensive due to the high non-linearity and non-convexity of neural networks. 
It has been proved to be NP-complete even for the simple fully-connected  networks with the piece-wise linear activation function ReLU  \cite{katz2017reluplex}. It is  significantly more difficult for those networks that contain differentiable S-curve  activation functions such as Sigmoid, Tanh, and Arctan \cite{zhang2018efficient}. 
To improve scalability, a practical solution is to over-approximate the nonlinear activation functions using linear upper and lower bounds. The verification problem is then transformed into an efficiently solvable linear programming problem. The linear over-approximation is a  prerequisite for other advanced verification approaches based on abstraction  \cite{singh2019abstract,pulina2010abstraction,elboher2020abstraction},  interval bound propagation (IBP) \cite{huang2019achieving}, and convex optimization \cite{wong2018provable,salman2019convex}.

As over-approximations inevitably introduce overestimation, the corresponding verification approaches  sacrifice completeness and may fail to prove or disprove the robustness of a neural network  \cite{liu2021algorithms}. Consequently, we cannot conclude that a neural network is not robust when we fail to prove it is robust by over-approximation. To resolve such uncertainties, an ideal approximation must be as tight as possible. Intuitively, an approximation is tighter if it introduces less overestimation.   

Considerable efforts have been devoted to finding tight over-approximations for precise verification results \cite{zhang2018efficient,lee2020lipschitz,wu2021tightening,lin2019robustness,DBLP:conf/nips/TjandraatmadjaA20}. 
Unfortunately, most of the tightness definitions lack  theoretical guarantees that a tighter approximation always implies a more precise verification result. Recent work has shown that none of them is superior to the others in terms of the verification results \cite{2208.09872}.
Lyu \textit{et al.} \cite{lyu2020fastened} and Zhang \textit{et al.}~\cite{2208.09872} claim that computing the tightest approximation is essentially a network-wise non-convex optimization problem, which is almost impractical to solve directly due to high computational complexity. They show that the scalability is significantly reduced when resorting to optimization. Hence, computing a tight approximation to each individual activation function is still an effective and practical solution. As  existing tightness characterizations of neuron-wise over-approximations cannot guarantee the  tightness, it is highly desirable to explore the missing factors on defining tighter neuron-wise approximations. 

In this paper we report a new crucial factor for defining tight over-approximations, namely the \emph{ approximation domains} of activation functions. It is  overlooked by existing approaches which consider only  over-estimated domains of the activation functions to approximate. However, an over-approximation that is tight on the overestimated domain may not be tight on the actual domain of the approximated function. 
There can be a tighter approximation of the actual domain.
 Unfortunately, computing the actual domain of the activation function on each neuron in a neural network is as difficult as the verification problem and thus impractical. 




Inspired by our new finding, we propose a novel under-approximation-guided over-approximation approach to define tight linear approximations for the robustness verification of neural networks. 
More specifically, we leverage the under-approximation to compute an underestimated domain for the approximated activation function. Moreover, we use both the underestimated domain and the overestimated domain to define a tight linear approximation for the function. We call it a \textit{dual-approximation} approach. We propose two under-approximation approaches, i.e.,  \textit{sampling-based} and  \textit{gradient-based}, to compute underestimated domains. The former is more time-wise efficient, while the latter is more precise in terms of verification results.
Our dual-approximation approach can produce tighter linear approximations than existing single-approximation approaches. 
We evaluate our approach on a comprehensive benchmark of 84 neural networks. 
The experimental results show that our approach outperforms all existing tight approaches, with up to 71.22\% improvement to the verification result.


In summary, we make three main contributions:
\begin{enumerate}[(1)]
	
	\item Identifying an crucial factor, namely approximation domain, in defining tight over-approximations for neural network robustness verification. 
	
	\item Presenting the first dual-approximation approach, i.e., under-approximation-guided  over-approximation, which outperforms the state-of-the-art approaches with up to 71.22\% improvement to  the robustness verification results. 
	
	\item Proposing two under-approximation approaches, i.e., gradient-descent-based and sampling-based, for computing under-approximated domains and experimentally demonstrating that the latter has a better performance than the former in the verification. 
	
\end{enumerate}

%


The remainder of this paper is organized as follows: Section \ref{sec:prel} gives preliminaries. Section \ref{sec:inter} discusses the interdependency between approximation domains and the tightness of approximations. We present our dual-approximation approach in Section \ref{sec:approx} and two under-approximation approaches in Section \ref{sec:under}, respectively. Section \ref{sec:exp} shows the experimental results. We discuss related work in Section \ref{sec:rel} and conclude the paper in Section \ref{sec:conc}.

	\section{Preliminaries}\label{sec:prel}
This section introduces the basic notation and terminology used throughout the paper. 
\vspace{-1mm}	 
\subsection{Deep Neural Networks}
\vspace{-1mm}
A deep neural network is a network of neurons as shown in Fig. \ref{fig:dnn}, which implements a mathematical function $F:\mathbb{R}^n \rightarrow \mathbb{R}^{m}$, e.g., $n=3$ and $m=2$ for the 2-hidden-layer DNN in Fig. \ref{fig:dnn}. Neurons except input ones 
are also functions $f:\mathbb{R}\rightarrow \mathbb{R}$ in the form of $f(x)=\sigma(Wx+b)$, where $\sigma(\cdot)$ is called an \textit{activation function}, $W$ a matrix of weights and $b$ a bias. A vector of $n$ numbers 
are fed into the network from the \textit{input layer} and propagated layer by layer through the internal \emph{hidden layers} after being multiplied by the weights on the edges, summed at the successor neurons with the bias and then computed by the neurons using the activation functions. 
The neurons on the \textit{output layer} compute the probabilities of classifying  an input vector to the labels represented by the corresponding neurons. The vector can be an image, a sentence, a voice, or a system state, depending on the application domains of the networks.

Given an $l$-layer neural network, let $k_i$ be the number of the neurons on the $i$-th layer, $W^{(i)}$ be the matrix of weights between the $i$-th and $(i+1)$-th layers, and $b^{(i)}$ the biases on the corresponding neurons, where $i=1,\ldots,l-1$. The function $F:\mathbb{R}^n \rightarrow \mathbb{R}^{m}$ implemented by the network can be defined by:  
\begin{align}
	&F(x)=W^{(l-1)}\sigma(z^{(l-1)}(x)),\tag{Network Function} \\
	\text{where},~&z^{(i)}(x) = W^{(i)} \sigma(z^{(i-1)}(x)) + b^{(i)}\tag{Layer Function}
\end{align}
for $i=2,\ldots,l-1$. For the sake of 
simplicity, we use $\hat{z}^{(i)}(x)$ to denote $\sigma(z^{(i)}(x))$ and $\Phi(x)= \arg \max_{\ell\in L} F(x)$ to denote the label $\ell$ such that the probability $F_{\ell}(x)$ of classifying $x$ to $\ell$ is larger than those to other labels, where $L$ represents the set of labels. 
The activation function $\sigma$ usually can be a Rectified Linear Unit (ReLU), $\sigma(x)=max(x,0)$), a Sigmoid function $\sigma(x)=\frac{1}{1+e^{-x}}$, a Tanh function $\sigma(x) = \frac{e^x - e^{-x}}{e^x + e^{-x}}$, and an Arctan function $\sigma(x) = tan^{-1}(x)$. In this work, we focus on the networks with S-curve activation functions, i.e., Sigmoid, Tanh, and Arctan, because they heavily rely on approximations.


Given a training dataset, the task of training a DNN is to fine-tune the weights and biases so that after the trained DNN achieves desired precision on test sets. 
Although a DNN is a precise mathematical function, its correctness is very challenging to guarantee due to the lack of formal specifications and the inexplicability of neural networks. Unlike programmer-composed programs, the machine-trained models are almost impossible to assign semantics to the internal computations.



\begin{figure}
	\begin{center}
		\includegraphics[width=0.46\textwidth]{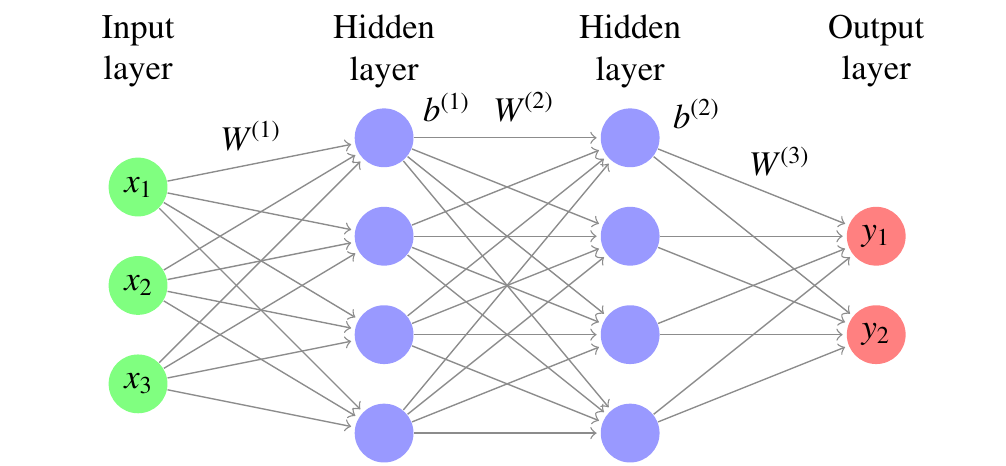}
	\end{center}
	\vspace{-2mm}
	\caption{A 4-layer DNN with two hidden layers.}
	\label{fig:dnn}
	\vspace{-4mm}
\end{figure}

\subsection{Robustness and Robustness Verification}
Despite the challenge in verifying the correctness of DNNs, formal verification is still useful to verify their safety-critical properties. One of the most important properties is  \emph{robustness}, stating that a neural network can classify an input correctly even if the input is manipulated under a reasonable range.

\begin{definition}[Local Robustness]\label{robust_def}
	A neural network $F:\mathbb{R}^n \rightarrow \mathbb{R}^{m}$ is called \textit{locally robust} with respect to an input $x_0$ and an input region $\Omega$ around $x_0$ if $\forall x \in \Omega, \Phi(x) = \Phi(x_0)$ holds.
\end{definition}

Usually, input region $\Omega$ around input $x_0$ is defined by a $\ell_p$-norm ball around $x_0$ with a radius of $\epsilon$, i.e. $\mathbb{B}_p (x_0, \epsilon):={x | \| x-x_0\|_p \le \epsilon}$. In this paper, we focus on the infinity norm and verify the robustness of the neural network in $\mathbb{B}_\infty (x_0, \epsilon)$. A corresponding robust verification problem is to compute the largest $\epsilon_0$ s.t. neural network $F$ is robust in $\mathbb{B}_\infty (x_0, \epsilon_0)$. The largest $\epsilon$ is called a \textit{certified lower bound}, which is a metric for measuring both the robustness of neural networks and the precision of robustness verification approaches. 

Assuming that the output label of $x_0$ is $c$, i.e. $\Phi(x_0) = c$, proving $F$'s robustness in Definition \ref{robust_def} is equivalent to proving that: $\forall x \in \Omega \forall \ell\in L/\{c\}.F_c(x) - F_\ell(x) > 0$ holds. 
Without the loss of generality, in this paper we fix some $\ell$ in $ L/\{c\}$ and 
verify that an input $x_0$ cannot be misclassified to label $\ell$ when $x_0$ is manipulated in the region $\Omega$. The verification problem is equivalent to solving the following optimization problem:
\begin{equation}\label{verify_problem}
	\min F_c(x) - F_{\ell}(x)
\end{equation}
We can conclude that $F$ is robust in $\Omega$ if the result is positive.  Otherwise, 
there exists some input $x'$ in $\Omega$ such that $F_\ell(x')\geq F_c(x')$. 
Namely, the probability of classifying $x'$ by $F$ to $\ell$ is greater than or equal to the one to $c$, and consequently, $x'$ 
may be classified as $\ell$ or $c$, meaning that $F$ is not robust in $\Omega$.

\begin{figure}
	\begin{center}
		\includegraphics[width=0.48\textwidth]{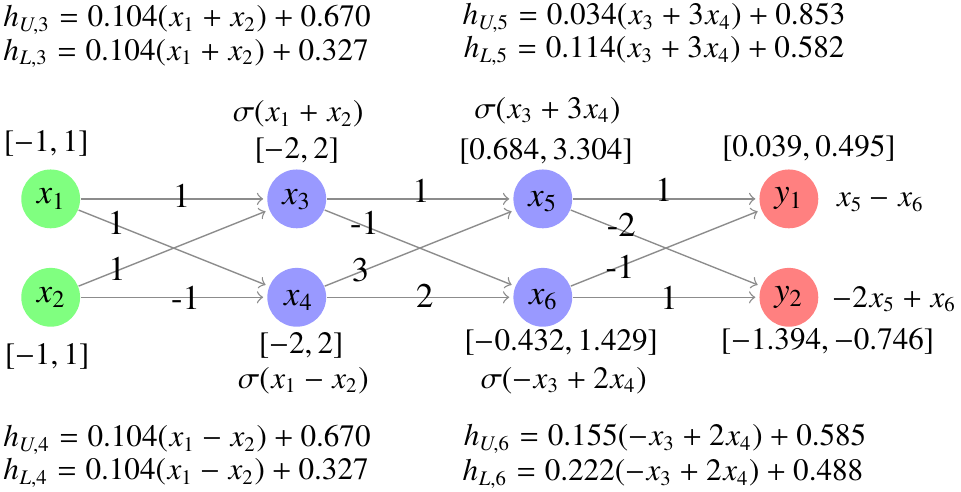}
		\caption{Verifying the robustness of a 4-layer network containing Sigmoid activation functions by over-approximations.}
		\label{approximation_eg}
	\end{center}
\vspace{-7mm}
\end{figure}


\subsection{Robustness Verification using Linear Over Approximation}
The optimization problem \ref{verify_problem} is computationally expensive, and it is almost impractical to compute the precise solution. The root reason for the high computational complexity of the problem is the non-linearity of the activation function $\sigma$. Even when $\sigma$ is piece-wise linear, e.g., the commonly used 
ReLU ($\sigma(x)=\max(x,0)$), the problem is still NP-complete \cite{katz2017reluplex}.



A practical solution to simplify the verification problem is to over-approximate $\sigma$ by linear constraints and transform it into an efficiently-solvable linear programming problem.


\begin{definition}[Linear Over-Approximation]\label{linear_bound_def}
Let $\sigma(x)$ be a non-linear function on $[l,u]$, there exist, $\alpha_L, \alpha_U, \beta_L, \beta_U \in \mathbb{R}$ \textit{s.t.}
	\begin{align*}
		& h_U(x) = \alpha_U x + \beta_U, \quad
		h_L(x) = \alpha_L x + \beta_L, and\\
		& \forall x \in [l,u], \quad h_L(x) \le \sigma(x) \le h_U(x).
	\end{align*}
	$h_L(x)$ and $h_U(x)$ are called the lower and upper linear bounds of $\sigma(x)$ on $[l,u]$, respectively, and $[l,u]$ is called the approximation domain of $\sigma(x)$ with respect to $h_L$ and $h_U$. 
\end{definition}

Using the lower and upper linear bounds defined in Definition \ref{linear_bound_def}, we can simplify  Problem (\ref{verify_problem}) to be the following efficiently solvable  linear optimization problem:
\begin{equation}\label{over-approx}
	\begin{split}
		\min & f(x) := F_c(x) - F_{t_0}(x) = z_c^{(m)}(x) - z_{t_0}^{(m)}(x) \\
		s.t. \quad & z^{(i)}(x) = W^{(i)} \hat{z}^{(i-1)}(x) + b^{(i)}, i \in {1,...,m}\\
		& h_L^{(i)}(x) \le \hat{z}^{(i)}(x) \le h_U^{(i)}(x), i \in {1,...,m-1}\\
		& x \in \Omega, t_0 \in L/c
	\end{split}
\end{equation}


\begin{example}\label{exm:1}
We consider an example in Figure \ref{approximation_eg}, which shows the verification of a simple neural network based on linear approximation. It is a fully-connected neural network with two hidden layers, $x_1, x_2 $ are input neurons and $y_1, y_2$ are output neurons. The intervals  represent the range of neurons before the application of the activation function. We conduct linear bounds for each neuron with an activation function using the information of intervals. $h_{U,i}$ and $h_{L,i}$ are the upper and lower linear bounds of $\mathit{\sigma}(x_i)$ respectively. From the computed intervals of output neuron, we have $y_1-y_2 > 0$ for all the possible $(x_1,x_2)$ in the input domain $[-1,1]\times [-1,1]$. As a result, we can conclude that the network is robust in the input domain with respect to the class corresponding to $y_1$.
\end{example}

The approximations inevitably introduce the overestimation of output ranges. 
In Example \ref{exm:1}, the real output range of $y_1$ is $[0.127, 0.392]$, which is computed by solving the optimization problems of minimizing and maximizing $y_1$, respectively. 
The one computed by over-approximation is $[0.039,0.495]$. We use the length difference of intervals to measure the overestimation. 
The overestimation can be as large as 72.08\%, even for a simple network.

The overestimation introduced by over-approximations may cause an actually robust case that cannot be verified, which is known as \emph{incomplete}. 
For instance, when we have $y_1-y_2<0$ by solving Problem \ref{over-approx}, there are two possible reasons. One is that there  exists some input such that the output on $y_1$ is indeed less than the one on $y_2$. The other is that the overestimation of $y_1$ and $y_2$ causes inequality. The network is robust in the latter case and not in the former. In this case, the algorithms just simply report \emph{unknown} because they cannot determine which reason causes it.

Note that an activation function can be approximated more tightly using two more pieces of linear bounds. However, 
the one-piece linear approximation in Definition \ref{linear_bound_def} is the most efficient because the reduced problem is a linear programming problem that can be efficiently solved in polynomial time. For piece-wise approximations, the number of linear bounds drastically blows up, and the corresponding reduced problem has been proved to be NP-complete \cite{singh2019abstract}.

\subsection{Approximation Tightness}
Reducing the overestimation of approximations is the key to reducing failure cases. The precision of approximations are characterized by the notion of \emph{tightness} \cite{2208.09872}. Many efforts have been made to define the tightest possible approximations. The tightness definitions can be classified into \textit{neuron-wise} and \textit{network-wise} categories.

\subsubsection{Neuron-wise Tightness}	
The tightness of activation functions' approximations can be measured independently. Given two upper bounds $h_{U}(x)$ and $h'_{U}(x)$ of activation function $\sigma(x)$ on the interval $[l,u]$, $h_{U}(x)$ is apparently tighter than $h'_{U}(x)$ if $h_{U}(x)<h'_{U}(x)$ for any $x$ in $[l,u]$ \cite{lyu2020fastened}. However, when $h_{U}(x)$ and $h'_{U}(x)$ intersect between $[l,u]$, their tightness becomes non-comparable. 
Another neuron-wise tightness metric is the area size of the gap between the bound and the activation function, i.e., $\int_{l}^{u}(h_{U}(x)-\sigma(x))dx$. 
A smaller area implies a tighter approximation \cite{HenriksenL20,muller2022prima}. 
However, another metric is the output range of the linear bounds. 
An approximation is considered to be \textit{the tightest} if it preserves the same output range as the activation function \cite{2208.09872}. 

\subsubsection{Network-wise Tightness}
Recent studies have shown that neuron-wise tightness does not always guarantee that the compound of all the approximations of the activation functions in a network is tight too \cite{2208.09872}. 
This finding explains why the approaches based on the neuron-wise tightness metrics achieve the best verification results only for certain networks. This finding inspires new approximation approaches that consider multiple and even all the activation functions in a network to approximate simultaneously. The compound of all the activation functions' approximations is called the network-wise tightest with respect to an output neuron if the neuron's output range is the precisest. 

Unfortunately, finding the  network-wise tightest approximation has been proved a non-convex optimization problem, and thus computationally expensive \cite{lyu2020fastened,2208.09872}. From a pragmatic viewpoint, a neuron-wise tight approximation is useful if all the neurons' composition is also network-wise tight.  The work \cite{2208.09872} shows that there exists such an neuron-wise tight approach under certain constraints when the networks are monotonic. That approach does not guarantee to be  optimal when the neural networks contain both positive and negative weights.

    	\section{Approximation Domain and Dual Overestimation Effects}\label{sec:inter}
In this section, we show that there is another important factor, i.e., the \emph{approximation domain} of activation functions, in defining  neuron-wise tight approximations. We show that a more precise approximation domain is important to define a tighter approximation of an activation function. Based on this finding, we further identify that the over-approximations of activation functions have dual overestimation effects resulting in imprecise verification results. 



\subsection{Approximation Domain: A Missing Factor}
Although an activation function $\sigma(x)$ is defined on the entire real number set, its inputs only come from a small segment. That is because all the perturbed inputs to the hosting neural network are bounded, and the network is Lipschitz-continuous \cite{lee2020lipschitz,combettes2020lipschitz}. We use $[l,u]$ to denote the real input interval of $\sigma(x)$. As shown in Figure \ref{approximation_eg}, the real input intervals of the activation functions on the neurons $x_3,x_4$ are $[-2,2]$. When approximating these functions, we only need to consider their real input intervals. Unfortunately, computing  the real input interval of $\sigma(x)$ is computationally expensive, at least as difficult as the original robustness verification problem. It is essentially an optimization problem. 
\begin{definition}\label{actual_domain_def}
	Given a neural network $F$ and an input region 
	$\mathbb{B}_\infty(x_0, \epsilon)$, the actual domain of the $r$-th hidden neuron in the $i$-th layer is $[l_r^{(i)}, u_r^{(i)}]$, where, 
	\begin{equation*}
		\begin{split}
			l_r^{(i)} =&  \min z_r^{(i)}(x), u_r^{(i)} = \max  z_r^{(i)}(x)\\
			s.t. \quad & z^{(j)}(x) = W^{(j)}\hat{z}^{j-1}(x) + b^{(j)}, j \in {1,...,i}\\
			& \hat{z}^{(j)}(x) = \sigma(z^{(j)}(x)), j \in {1,...,i-1}\\
			& x \in \mathbb{B}_\infty(x_0, \epsilon).
		\end{split}
	\end{equation*} 
\end{definition}

To avoid the complexity of computing actual domains of activation  functions, almost all the existing approaches define approximations on  estimated domains, which should be overestimated in order to guarantee the soundness of the approximations. Let's consider the activation functions on neurons $x_5$ and $x_6$ in Figure \ref{approximation_eg} as an example. Their input intervals are estimated based on the approximations of $x_3$ and $x_4$ by solving the corresponding linear programming problems. The input intervals are  $[0.684,3.304]$ and $[-0.432,1.429]$, respectively. 
Compared with the real ones, they are overestimated, as shown in Table \ref{range_compare}. Apparently, an over-approximation defined on an overestimated input interval is also a valid over-approximation on the real input interval based on Definition \ref{linear_bound_def}. 



%
%
%
%
%
%
%
%
%
%
%
%
%

\begin{figure}
	\begin{center}
		\includegraphics[width=.28\textwidth]{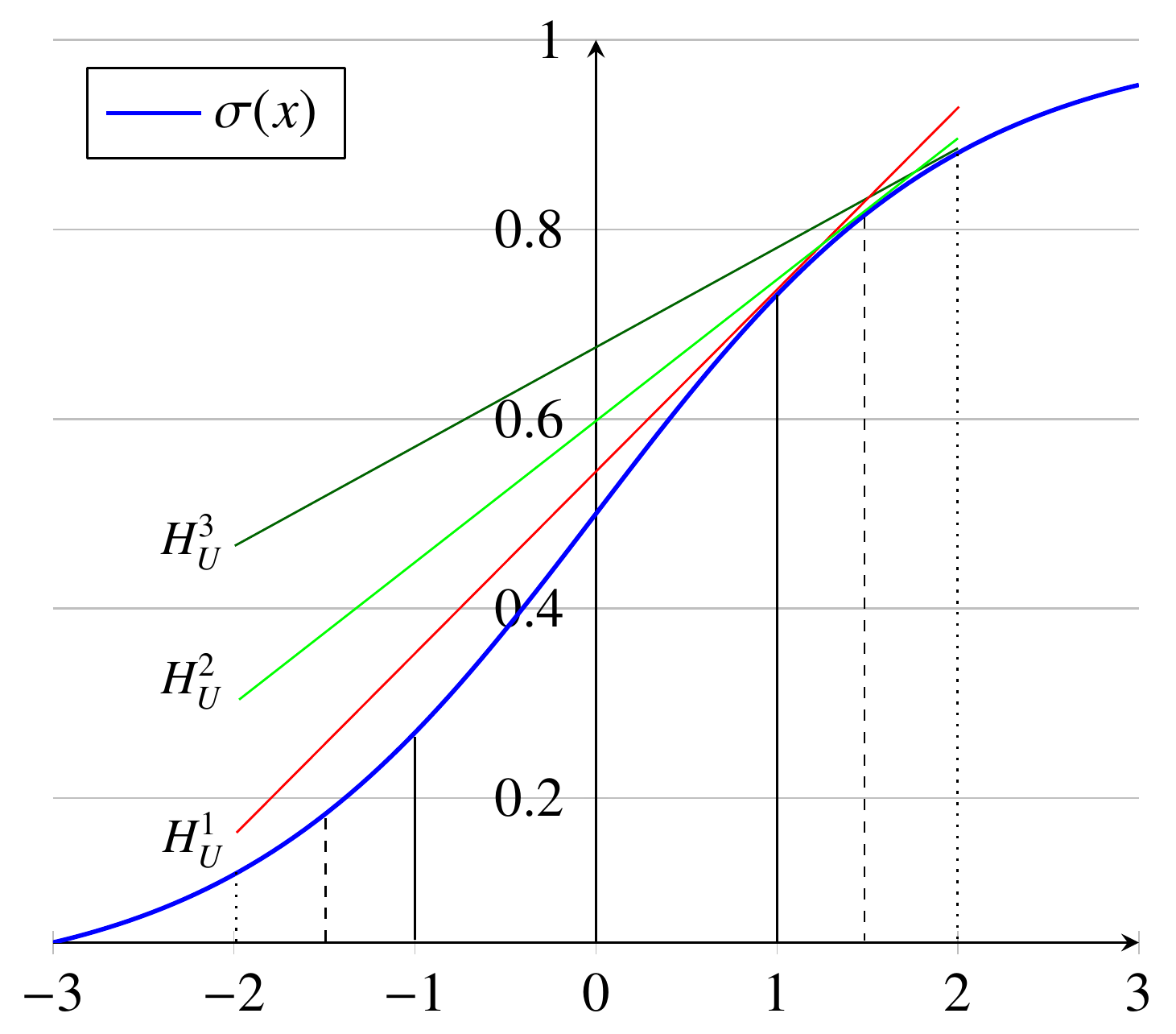}
		\caption{An example of defining the neuron-wise tightest linear upper bounds on different  approximation domains.}
		\label{fig:domain}
	\end{center}
\vspace{-7mm}
\end{figure}

The overestimation of approximation domains is the major factor in defining non-tight over-approximations. More overestimation results in less tightness. Figure \ref{fig:domain} depicts an example. We assume that the actual  domain of $\sigma(x)$ is $[-1,1]$ and the tangent line $H^1_{U}(x)$ at $(1,\sigma(1))$ is the tightest upper bound. Let $H^2_{U}(x)$ and 
$H^3_{U}(x)$ be the tightest linear upper bounds on the domains  $[-1.5,1.5]$ and $[-2,2]$, respectively. Apparently, $H^2_{U}(x)$ is tighter than $H^3_{U}(x)$ on the domain $[-1,1]$ because $H^2_{U}(x)<H^3_{U}(x)$ for all $x$ in $[-1,1]$.

Next, we formally prove that a precise approximation domain is necessary to define tight over-approximations. 



\begin{theorem}\label{thm:1}
	Given two overestimated approximation domains $[l_{\mathit{over}},u_{\mathit{over}}]$ and $[l'_{\mathit{over}},u'_{\mathit{over}}]$ satisfying  $l'_{\mathit{over}} < l_{\mathit{over}} \le l, u \le u_{\mathit{over}} < u'_{\mathit{over}}]$, For any pair of linear bounds $(h'_L(x), h'_U(x))$ of continuous function $\sigma$ defined on $[l'_{\mathit{over}},u'_{\mathit{over}}]$, there exists a pair of linear bounds $(h_L(x), h_U(x))$  defined on $[l_{\mathit{over}},u_{\mathit{over}}]$ satisfying   $\forall x \in [l,u], h'_L(x) \le h_L(x), h'_U(x) \ge h_U(x)$.
\end{theorem}

\begin{proof}
    Since we have $l'_{\mathit{over}} < l_{\mathit{over}}$ and $u_{\mathit{over}} < u'_{\mathit{over}}$, which means $[l_{\mathit{over}},u_{\mathit{over}}] \subset [l'_{\mathit{over}},u'_{\mathit{over}}]$, and $\sigma$ is continuous function, we can conclude that the range of $\sigma$ satisfies: $\sigma([l_{\mathit{over}},u_{\mathit{over}}]) \subset \sigma([l'_{\mathit{over}},u'_{\mathit{over}}])$. So, we have $\max\limits_{x \in [l'_{\mathit{over}},u'_{\mathit{over}}]} \sigma(x) \ge \max\limits_{x \in [l_{\mathit{over}},u_{\mathit{over}}]} \sigma(x)$. Let $\delta_u := \max\limits_{x \in [l'_{\mathit{over}},u'_{\mathit{over}}]} \sigma(x) - \max\limits_{x \in [l_{\mathit{over}},u_{\mathit{over}}]} \sigma(x)$. We have $\delta_u \ge 0$. Then, we define a linear bound $h_U(x)$ on $[l_{\mathit{over}},u_{\mathit{over}}]$ as $h_U(x) = h'_U(x) - \delta_u$.
    For all $x$ in $[l,u]$, there is $h'_U(x) \ge h_U(x)$. 
    
    Next, we show that $h_U(x)$ is a sound upper bound for $\sigma$ on $[l_{\mathit{over}},u_{\mathit{over}}]$. That is, for all $x$ in $ [l_{\mathit{over}},u_{\mathit{over}}]$, we have: 
    \begin{align*}
         h_U(x) - \sigma(x) 
        &\ge h_U(x) - \max\limits_{x' \in [l_{\mathit{over}},u_{\mathit{over}}]} \sigma(x') \\
        &=  h_U(x) - \max\limits_{x' \in [l'_{\mathit{over}},u'_{\mathit{over}}]} \sigma(x')\\
        &~~~ + (\max\limits_{x' \in [l'_{\mathit{over}},u'_{\mathit{over}}]} \sigma(x') - \max\limits_{x' \in [l_{\mathit{over}},u_{\mathit{over}}]} \sigma(x'))\\
        &=  h_U(x) + \delta_u - \max\limits_{x' \in [l'_{\mathit{over}},u'_{\mathit{over}}]} \sigma(x')\\
        &=  h'_U(x) - \max\limits_{x' \in [l'_{\mathit{over}},u'_{\mathit{over}}]} \sigma(x') \ge 0.
    \end{align*}
    By Definition \ref{linear_bound_def}, $h_U(x)$ is a sound upper bound. Likewise, let $\delta_l := \min\limits_{x \in [l_{\mathit{over}},u_{\mathit{over}}]} \sigma(x) - \min\limits_{x \in [l'_{\mathit{over}},u'_{\mathit{over}}]} \sigma(x)$, then $h_L(x) := h'_L(x) + \delta_l$ is an eligible sound lower bound for $\sigma$ on $[l_{\mathit{over}},u_{\mathit{over}}]$.
\end{proof}

Theorem \ref{thm:1} indicates that the closer $[l_{\mathit{over}},u_{\mathit{over}}]$ gets to $[l,u]$, 
the tighter the approximation linear bounds can be defined. 
Therefore, we can define tight approximations by reducing the overestimation of the approximation domains.  


\begin{figure}
	\begin{center}
		\includegraphics[width=0.48\textwidth]{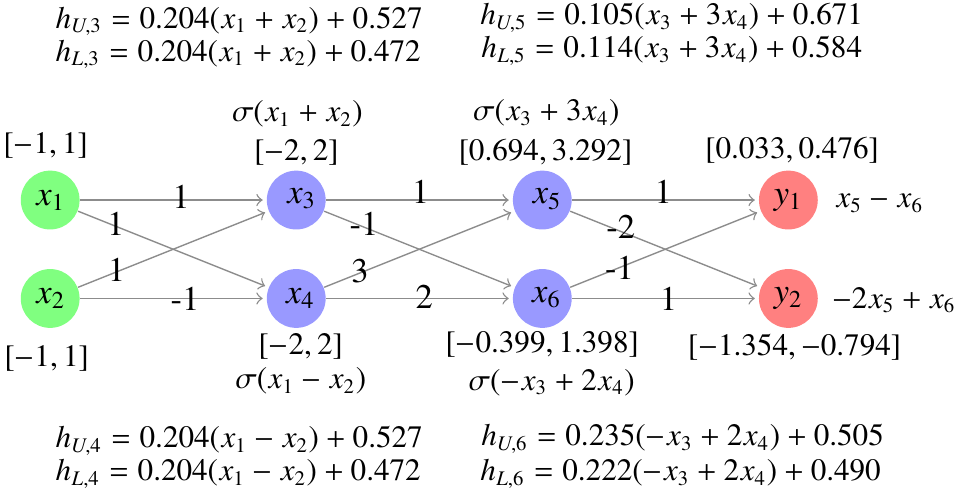}
		\caption{The tighter approximations for the activation functions in the same network in Figure \ref{approximation_eg}.}
		\label{approx_eg_tighter}
	\end{center}
	\vspace{-7mm}
\end{figure}

\vspace{-1mm}

\subsection{The Dual Overestimation Effects}
Overestimated approximation domains can be efficiently computed using layer-by-layer symbolic interval propagation \cite{dutta2018output, xiang2018output, wang2018formal}.
Specifically, we can overestimate an input domain $[l^{(i)}_{r,over},u^{(i)}_{r,over}]$ for the activation function $\sigma^{(i)}_r(x)$ on the $r$-th neuron in the $i$-the hidden layer by over-approximating each $\sigma(x)$ in Definition \ref{actual_domain_def} and transforming it into a linear optimization problem. 
Apparently, overestimation is introduced to the input domains of the activation functions. 
The left part of Table \ref{range_compare} shows the overestimation  
in Example \ref{exm:1}, where the overestimation  becomes larger and larger layer by layer.

According to Theorem \ref{thm:1}, the large overestimation of approximation domains would result in non-tight over-approximations to activation functions, which further leads to a larger overestimation of the approximation domains for the following activation functions. We call this phenomenon \emph{dual overestimation effects} of over-approximations.  Apparently, we need a tighter over-approximation of each $\sigma(x)$ to compute  a less overestimated input domain of $\sigma^{(i)}_r(x)$. 


%

\begin{proposition}\label{prop:right}
	A tighter over-approximation of an activation function yields less over-estimated approximation domains for its following activation functions. 
\end{proposition}
\vspace{-2mm}

\begin{example}\label{exm:2}
We revisit the network in Figure \ref{approximation_eg}. We define different over-approximations for the four activation functions as shown in Figure \ref{approx_eg_tighter} and compare the over-estimated input domains of $x_5,x_6$. 
Figure \ref{app-comp} visualizes the over-approximations, where the linear bounds in green are those in Figure \ref{approximation_eg} and the red ones represent those in Figure \ref{approx_eg_tighter}. 
The approximation on $x_5$ is similar to the one on $x_6$ and thus omitted due to space limit. 
In Figure \ref{approximation_eg}, we take the tangent lines at the two endpoints as the bounds \cite{2208.09872}, and in Figure \ref{approx_eg_tighter} we use the tangent lines across the endpoints \cite{wu2021tightening,HenriksenL20}. The overestimations are reduced by the tighter approximations defined in Figure \ref{approx_eg_tighter}. 

\begin{figure}
		\begin{subfigure}{0.24\textwidth}
		\includegraphics[width=\textwidth]{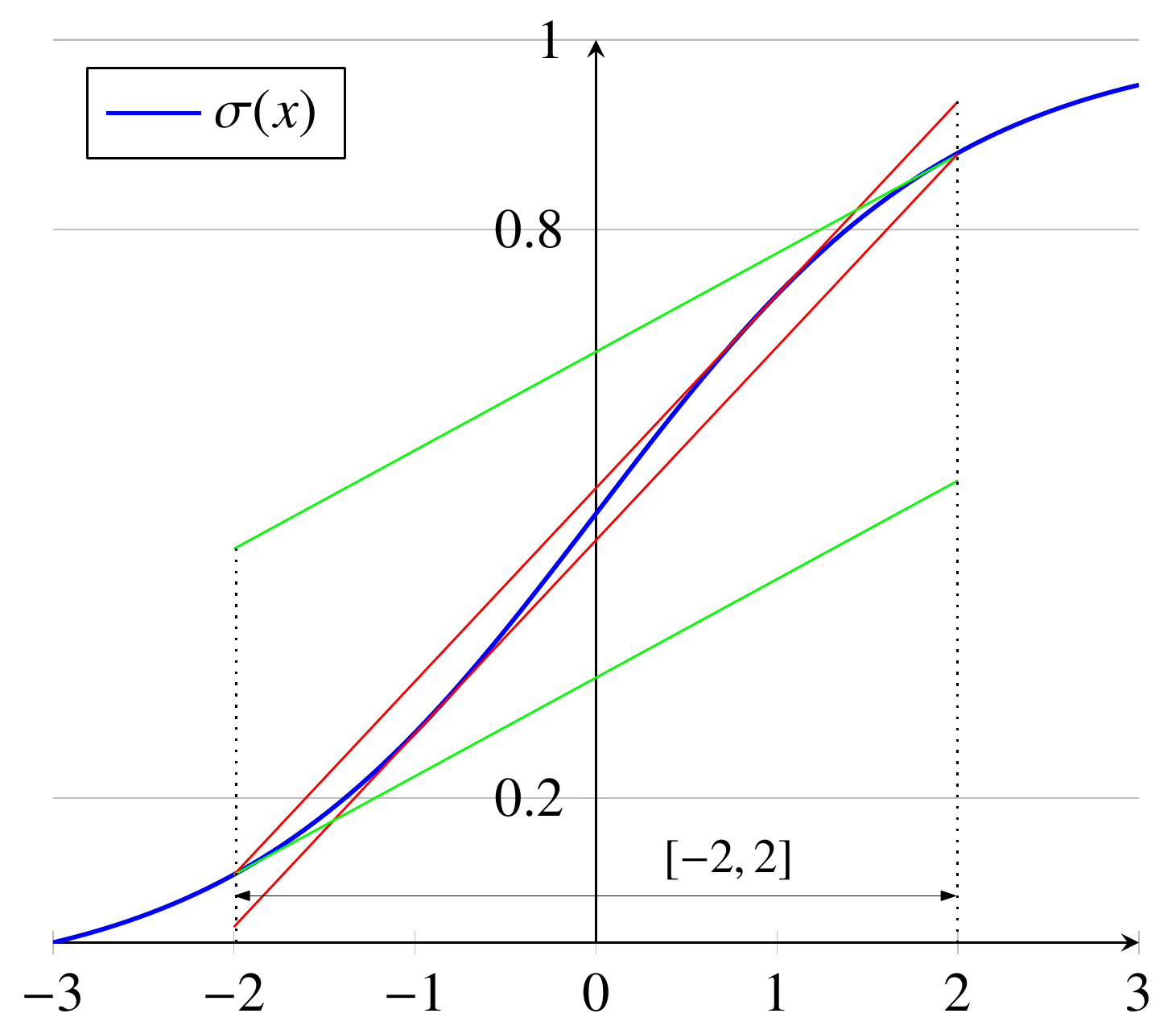}
		\caption{On neurons $x_3,x_4$.}
		\label{fig:app-x34}
	\end{subfigure}
		\begin{subfigure}{0.24\textwidth}
	\includegraphics[width=\textwidth]{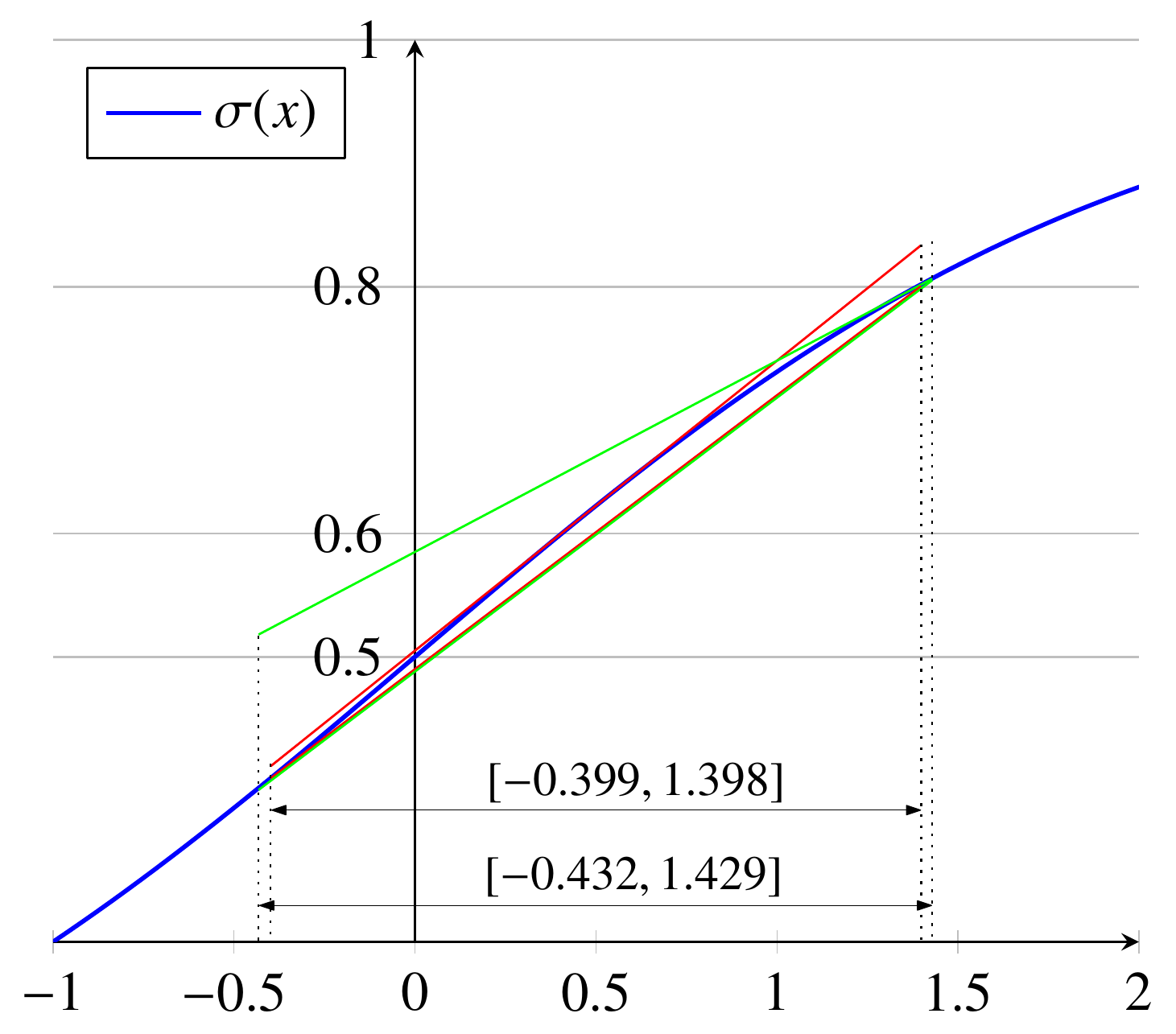}
	\caption{On neuron $x_6$.}
	\label{fig:app-x5}
\end{subfigure}
\vspace{-1mm}
\caption{Over-approximation comparisons in Figures \ref{approximation_eg} and \ref{approx_eg_tighter}.  }
\vspace{-1mm}
\label{app-comp}
\end{figure}


\end{example}

\begin{table}[]
	\centering
	\setlength{\tabcolsep}{2pt}	
	\caption{The overestimations that are introduced due to over-approximations.} 
	\vspace{-1mm}	
	\begin{tabular}{|l|r|r|r|r|r|}
		\hline
		&               & \multicolumn{2}{c|}{\textbf{Example 1}} & \multicolumn{2}{c|}{\textbf{Example 2}} \\\hline 
		 &\multicolumn{1}{c|}{\textbf{Real Domain}}    & \textbf{Overestimated}  & \textbf{Pct. (\%)} & \textbf{Overestimated}  & \textbf{Pct. (\%)} \\ \hline
		$x_3$ & [-2,2]&[-2,2]& 0&[-2,2] &0\\
		$x_4$ & [-2,2]&[-2,2]& 0&[-2,2] &0\\
		$x_5$ & {[}0.857, 3.142{]}   & {[}0.684, 3.304{]}   & 14.66   &{[}0.694, 3.292{]}&  13.70\\
		$x_6$ & {[}-0.266, 1.266{]}  & {[}-0.432, 1.429{]}  & 21.47   &{[}-0.399, 1.398{]}&  17.30\\
		$y_1$ & {[}0.127, 0.392{]}   & {[}0.039, 0.495{]}   & 72.08   &{[}0.033, 0.476{]}&  67.17\\
		$y_2$ & {[}-1.301, -0.924{]} & {[}-1.394, -0.746{]} & 71.88   &{[}-1.354, -0.794{]}&  48.54\\ \hline
	\end{tabular}
\vspace{-3mm}
	\label{range_compare}  
\end{table}

Example \ref{exm:2} shows that the over-approximation of activation functions has dual effects of causing the overestimation to increase drastically layer by layer. 
On the one hand, it causes overestimation of the output ranges of the neurons on the output layer. On the other hand, it overestimates the approximation domains of the activation functions following it.  
The overestimation of approximation domains further increases the overestimation when the corresponding activation functions are over-approximated on those domains. 


A solution to reducing the dual overestimation effects is to loosen the dependency between activation functions. One possible way of loosening the dependency is to define over-approximations based on some information that is independent of the approximations of predecessor activation functions, besides the overestimated approximation domain.  


\section{Under-Approximation Guided Over-Approximation}\label{sec:approx}
In this section, we propose an under-approximation-guided approach for defining tight over-approximations for activation functions. Specifically, for each activation function to approximate, we underestimate the approximation domain and leverage this information to guide the over-approximation. The underestimated domain is enclosed by the actual domain of the function. The underestimated domain guarantees its tightness, while the overestimated domain guarantees its soundness.

\subsection{Problem Definition}

Because computing the exact approximation domain is impractical and the overestimated approximation domain is not enough to define tight approximations, we consider introducing underestimated domains to approximate the exact ones and defining lower and upper bounds based on both of them. Specifically, given an S-curve function $\sigma(x)$, let $[l_{\mathit{over}},u_{\mathit{over}}]$ and $[l_{\mathit{under}},u_{\mathit{under}}]$ be an over-approximated domain and an under-approximated domain, respectively. That is, we have $l_{\mathit{over}}\leq l_x\leq l_{\mathit{under}}$ and $u_{\mathit{under}}\leq u_x\leq u_{\mathit{over}}$, where 
$l_x,u_x$ denote the lower and upper bounds of $x$. The problem is to define tight lower and upper linear bounds for $\sigma(x)$ on the domain $[l_x,u_x]$ based on 
$[l_{\mathit{under}},u_{\mathit{under}}]$ and $[l_{\mathit{over}},u_{\mathit{over}}]$. 

We take the case of the upper bound, for example.
\begin{theorem}\label{real_domain_tightness}
	For a monotonously increasing function $\sigma$ on $[l_{over},u_{over}]$ which is concave on $[u_x, u_{over}]$, an upper bound of $[l_x, u_x]$ defined by tangent line at the point $u_x$  must be tighter than the one at the point between $[u_x,u_{over}]$.  
\end{theorem}

\begin{proof} 
	Suppose that we have $d \in [u_x,u_{over}]$, then the tangent line at $u_x$ is $y = \sigma '(u_x) x + \sigma(u_x) -\sigma '(u_x) u_x $ and the one at $d$ is $y = \sigma '(d) x + \sigma(d) -\sigma '(d) d $. For any $x\in [l_{under},u_{under}]$, we consider $T(x) = \sigma '(u_x) x + \sigma(u_x) -\sigma '(u_x) u_x -(\sigma '(d) x + \sigma(d) -\sigma '(d) d)$, $x \in [l_x,u_x]$. Since $\sigma$ is concave on $[u_x, u_{over}]$, we have $\forall d \in [u_x, u_{over}]$, $ \sigma '(u_x) \ge \sigma '(d)$. So $T'(x) = \sigma '(u_x) - \sigma '(d) \ge 0$, which means that $T(x)$ is monotonously increasing on $[l_x,u_x]$. $T(u_x)= \sigma(u_x) -(\sigma '(d) u_x + \sigma(d) -\sigma '(d) d) = \sigma(u_x) - \sigma(d) + \sigma'(d) (d - u_x) \le 0$, for $\sigma$ is concave on $[u_x, u_{over}]$. As a result, $T(x) \le T(u_{under}) \le 0$, which means the tangent line at any point between $[u_x,u_{over}]$ is over the one at $u_x$ on $[l_x,u_x]$.
\end{proof} 

From the definitions of tightness of upper and lower bounds in \cite{lyu2020fastened}, we can conclude that considering the real domain $[l_x, u_x]$, the tangent line at point $u_x$ is tighter. We can extend Theorem \ref{real_domain_tightness} to lower bounds of monotonously decreasing function $\sigma$ on $[l_{over},u_{over}]$ which is convex on $[l_{under}, l_x]$, as it is the symmetry of Theorem \ref{real_domain_tightness}. From these conclusions, we can deduce that the tangent lines at $l_x$ and $u_x$ can guide us to find tighter bounds on the real domain of $x$. However, due to the complexity of computing, the precise $[l_x, u_x]$ is hard to find. We can only approach $[l_x, u_x]$ by computing $[l_{under}, u_{under}]$, and use the under-approximation domain for guidance of approximation.

\begin{figure*}
	\centering
	\begin{subfigure}{0.24\textwidth}
		\includegraphics[width=\textwidth]{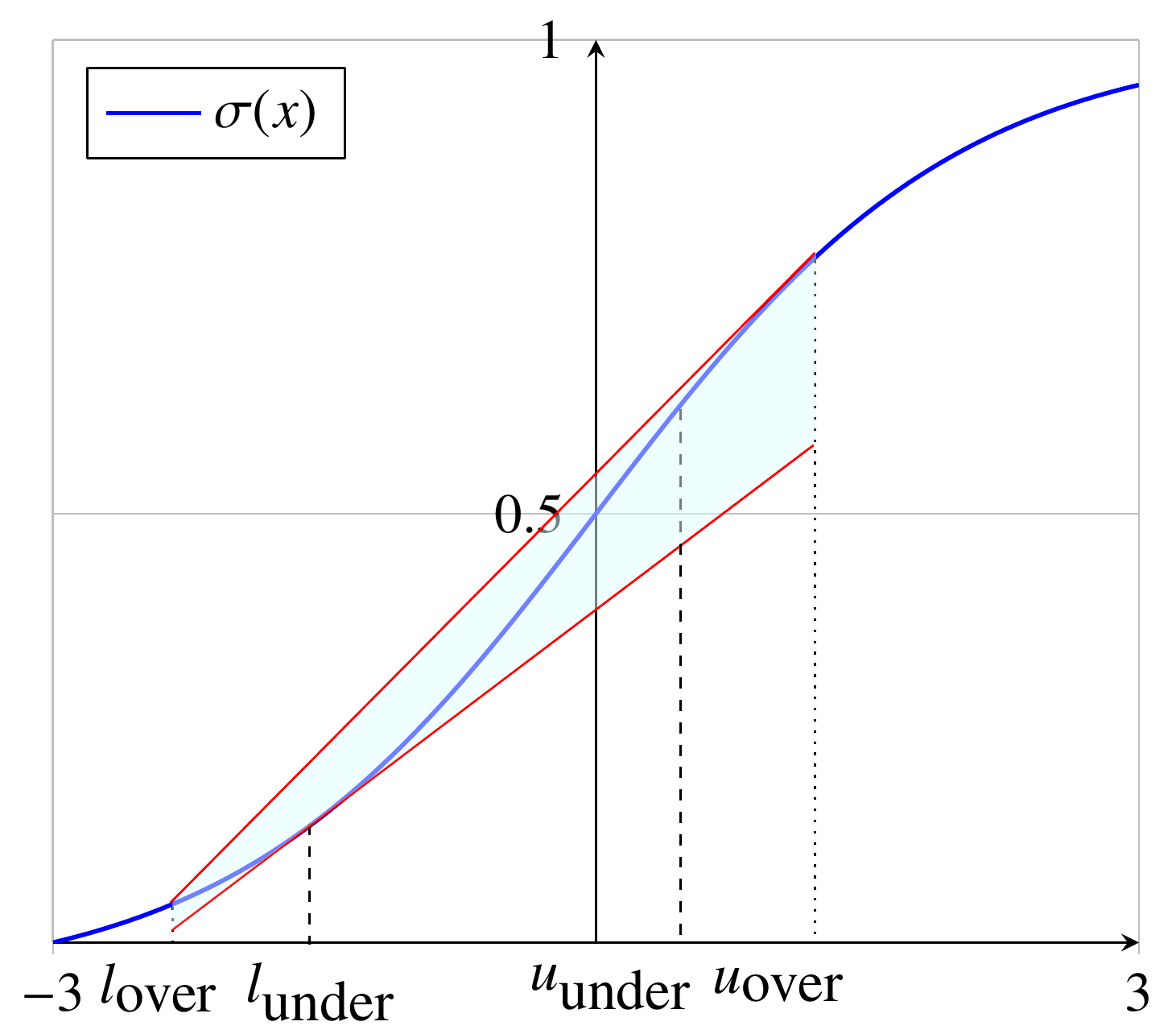}
		\caption{Case I: $\sigma'(l_U) < k < \sigma'(u_U)$.}
		\label{fig:a}
	\end{subfigure}
	\hfill
	\begin{subfigure}{0.24\textwidth}
		\includegraphics[width=\textwidth]{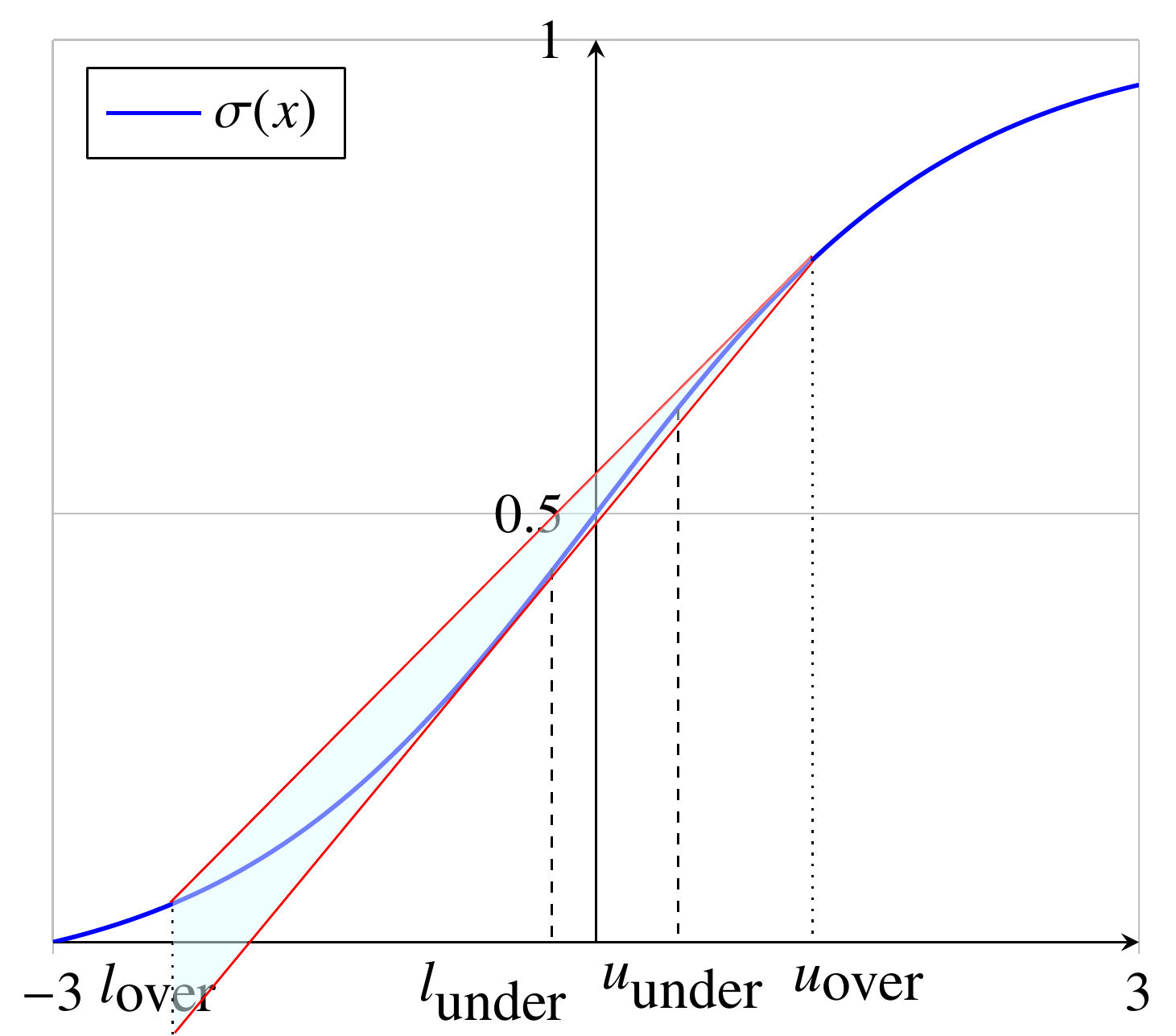}
		\caption{Case I: $\sigma'(l_U) < k < \sigma'(u_U)$}
		\label{fig:b}
	\end{subfigure}
	\hfill
	\begin{subfigure}{0.24\textwidth}
		\includegraphics[width=\textwidth]{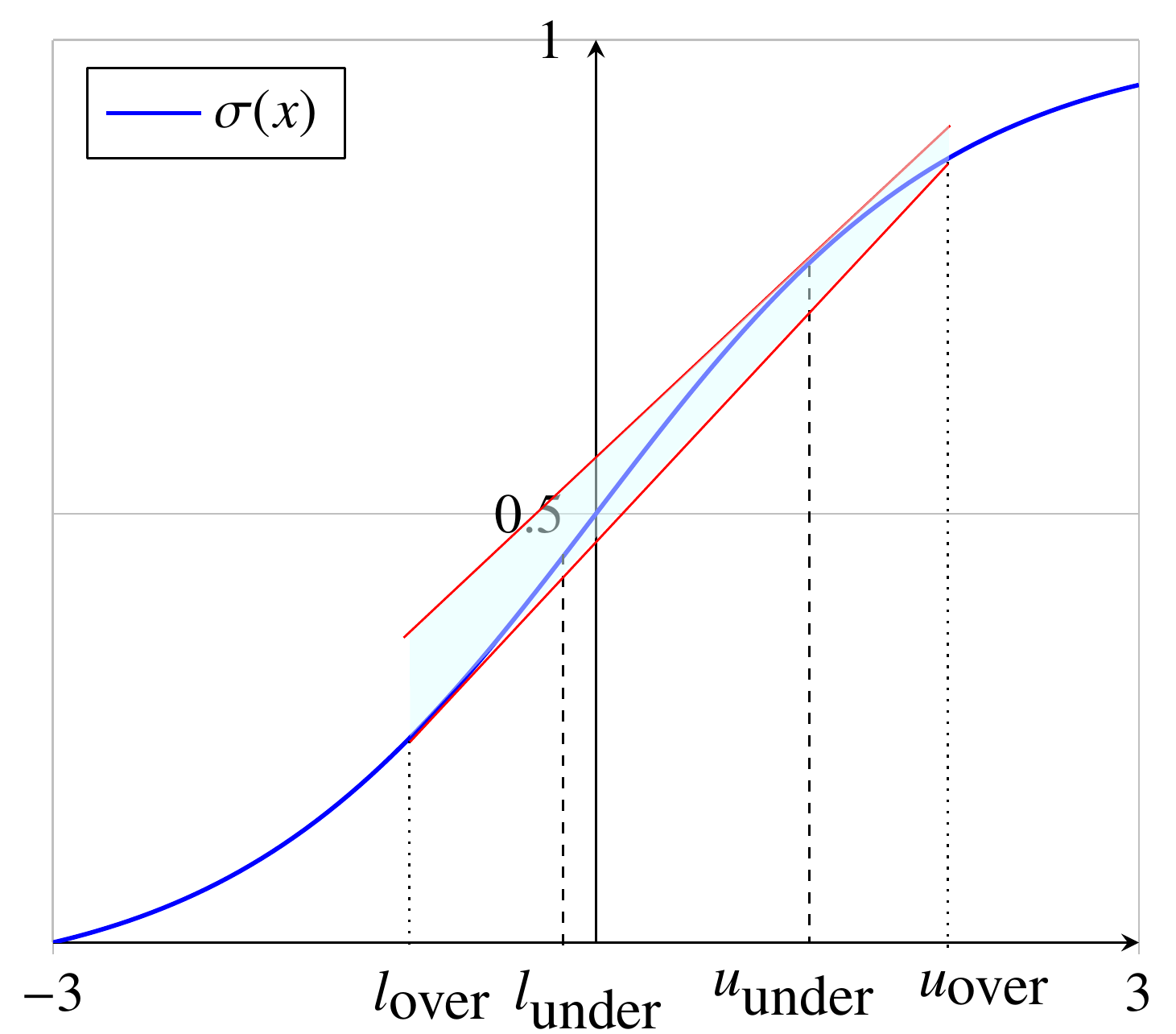}
		\caption{Case II: $\sigma'(l_U) > k > \sigma'(u_U)$.}
		\label{fig:c}
	\end{subfigure}
	\hfill
	\begin{subfigure}{0.24\textwidth}
		\includegraphics[width=\textwidth]{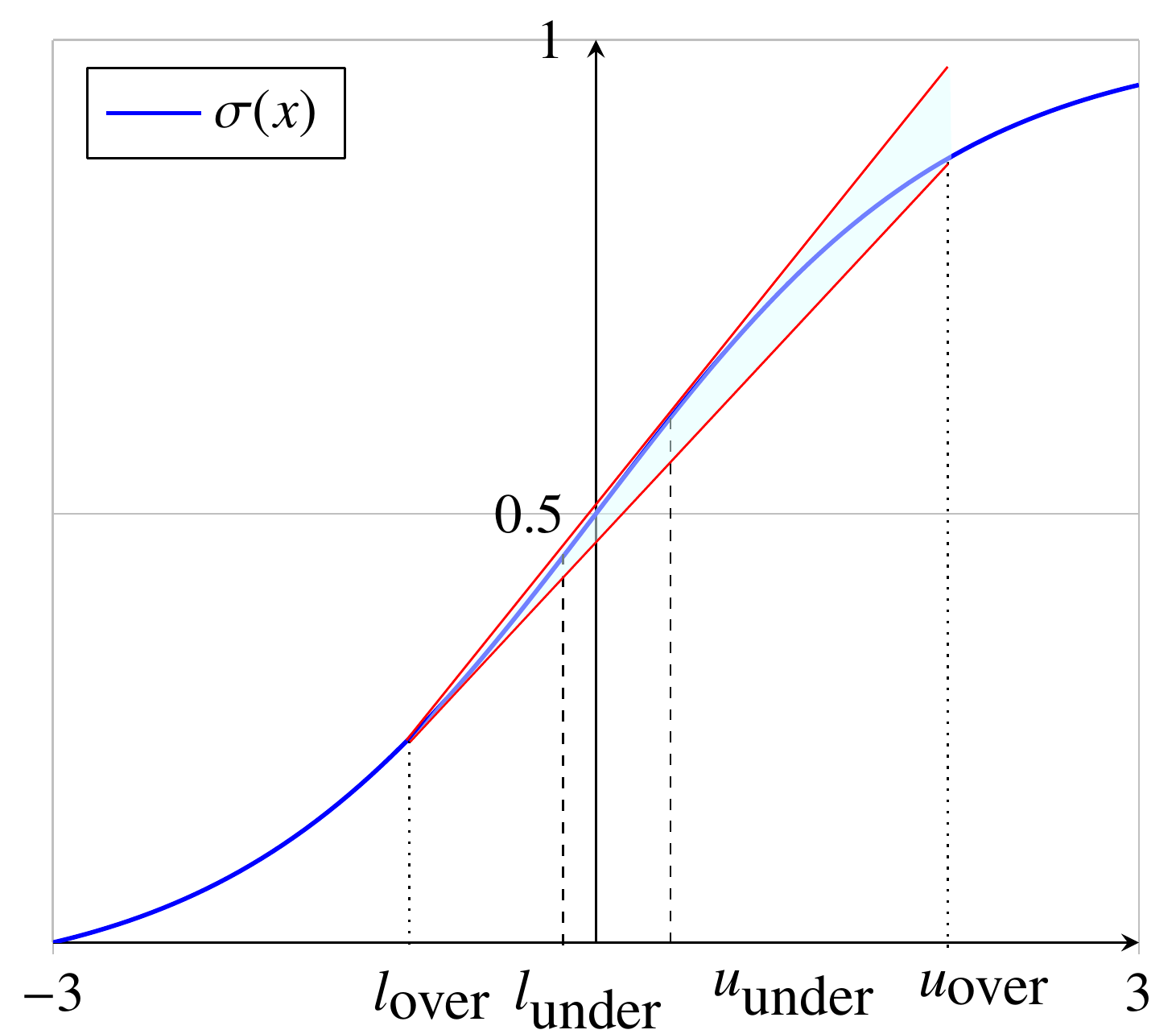}
		\caption{Case II: $\sigma'(l_U) > k > \sigma'(u_U)$.}
		\label{fig:d}
	\end{subfigure}
	\hfill\\
	\begin{subfigure}{0.24\textwidth}
		\includegraphics[width=\textwidth]{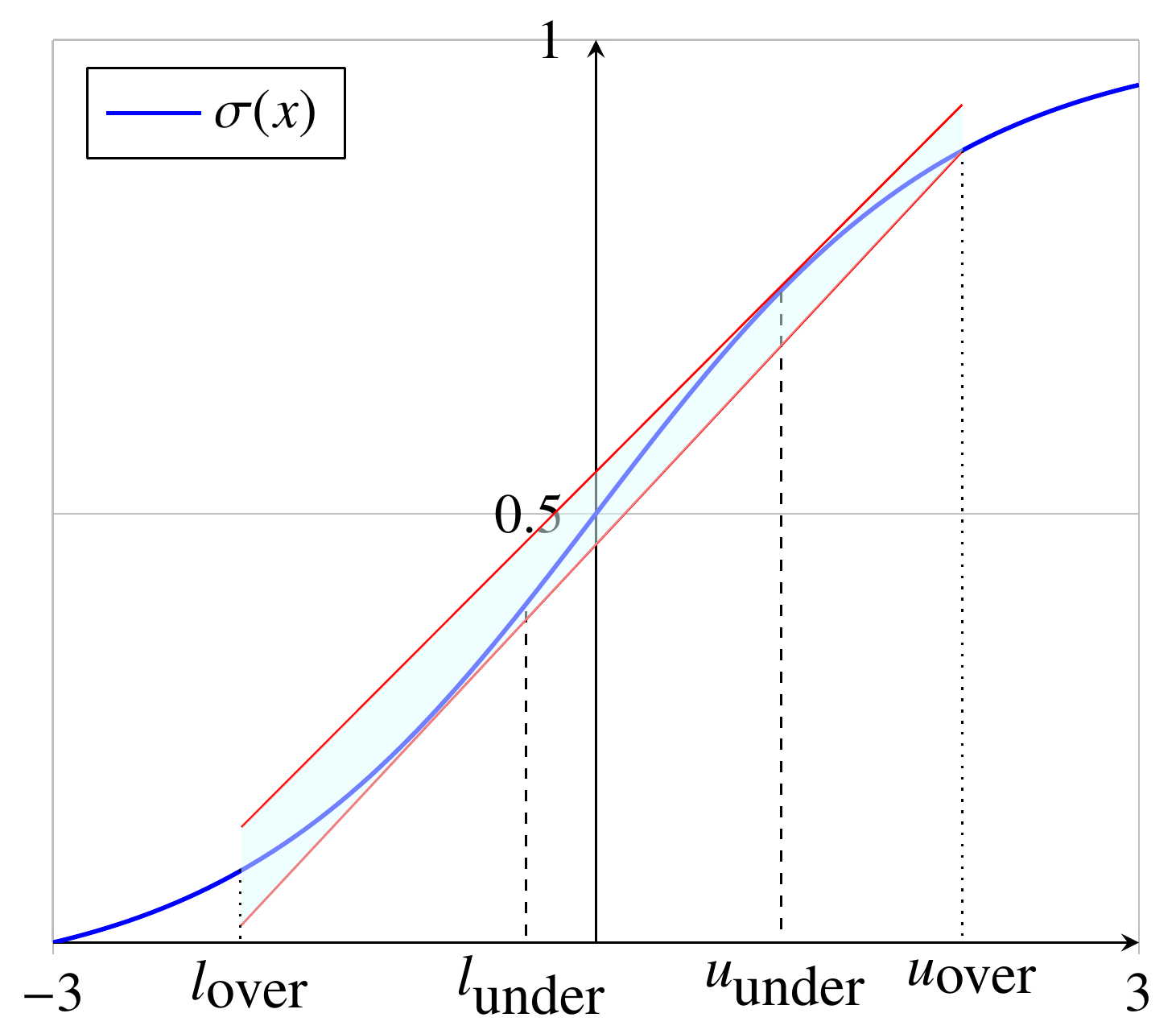}
		\caption{Case\,III: $\sigma'(l_U)\!<\!k\wedge\sigma'(u_U)\!<\!k$.}
		\label{fig:e}
	\end{subfigure}
	\begin{subfigure}{0.24\textwidth}
		\includegraphics[width=\textwidth]{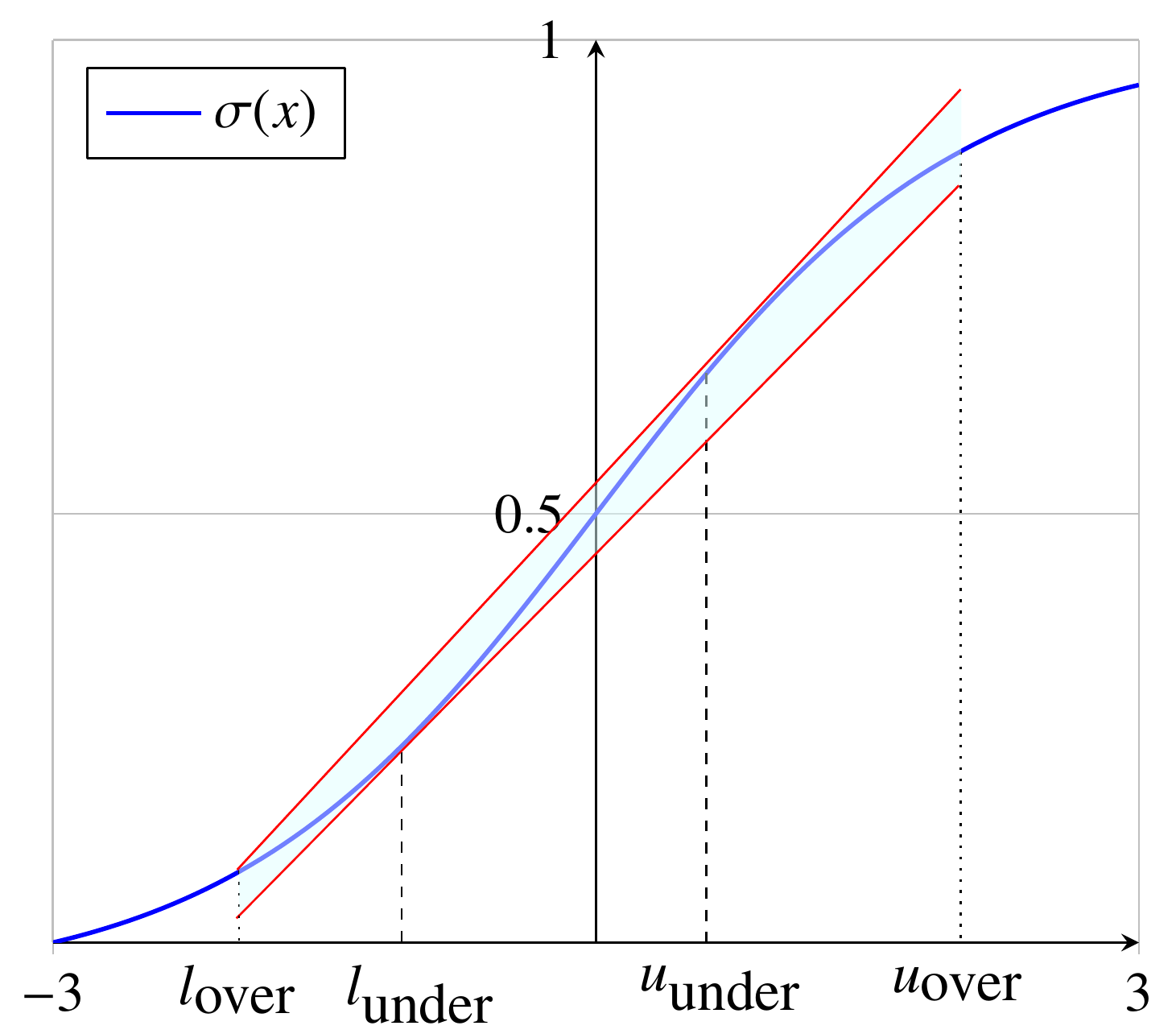}
		\caption{Case\,III: $\sigma'(l_U) \!<\! k\wedge\sigma'(u_U) \!<\! k$.}
		\label{fig:f}
	\end{subfigure}
	\hfill
	\begin{subfigure}{0.24\textwidth}
		\includegraphics[width=\textwidth]{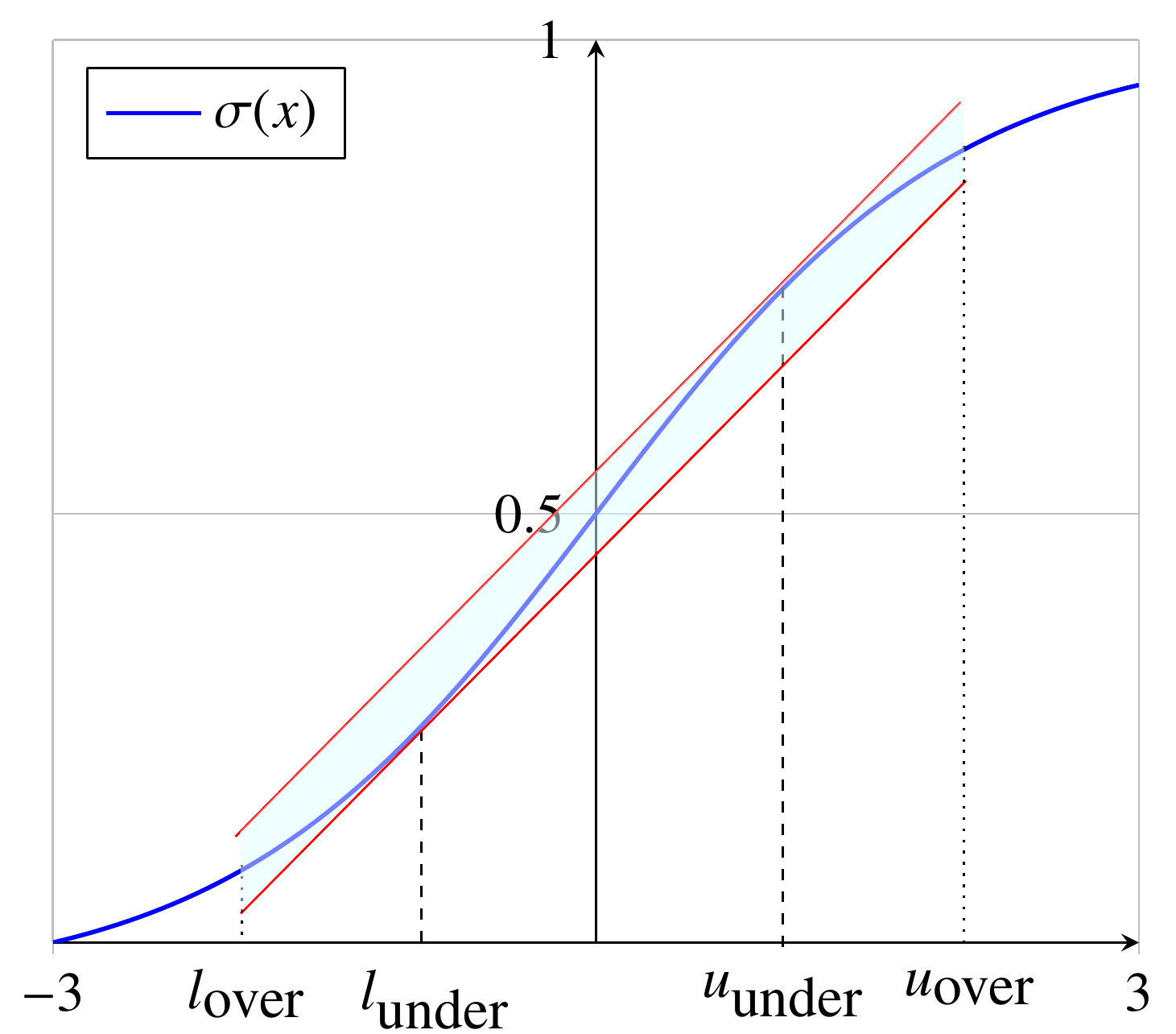}
		\caption{Case\,III: $\sigma'(l_U) \!<\! k\!\wedge\!\sigma'(u_U) \!<\! k$.}
		\label{fig:g}
	\end{subfigure}
	\hfill
	\begin{subfigure}{0.24\textwidth}
		\includegraphics[width=\textwidth]{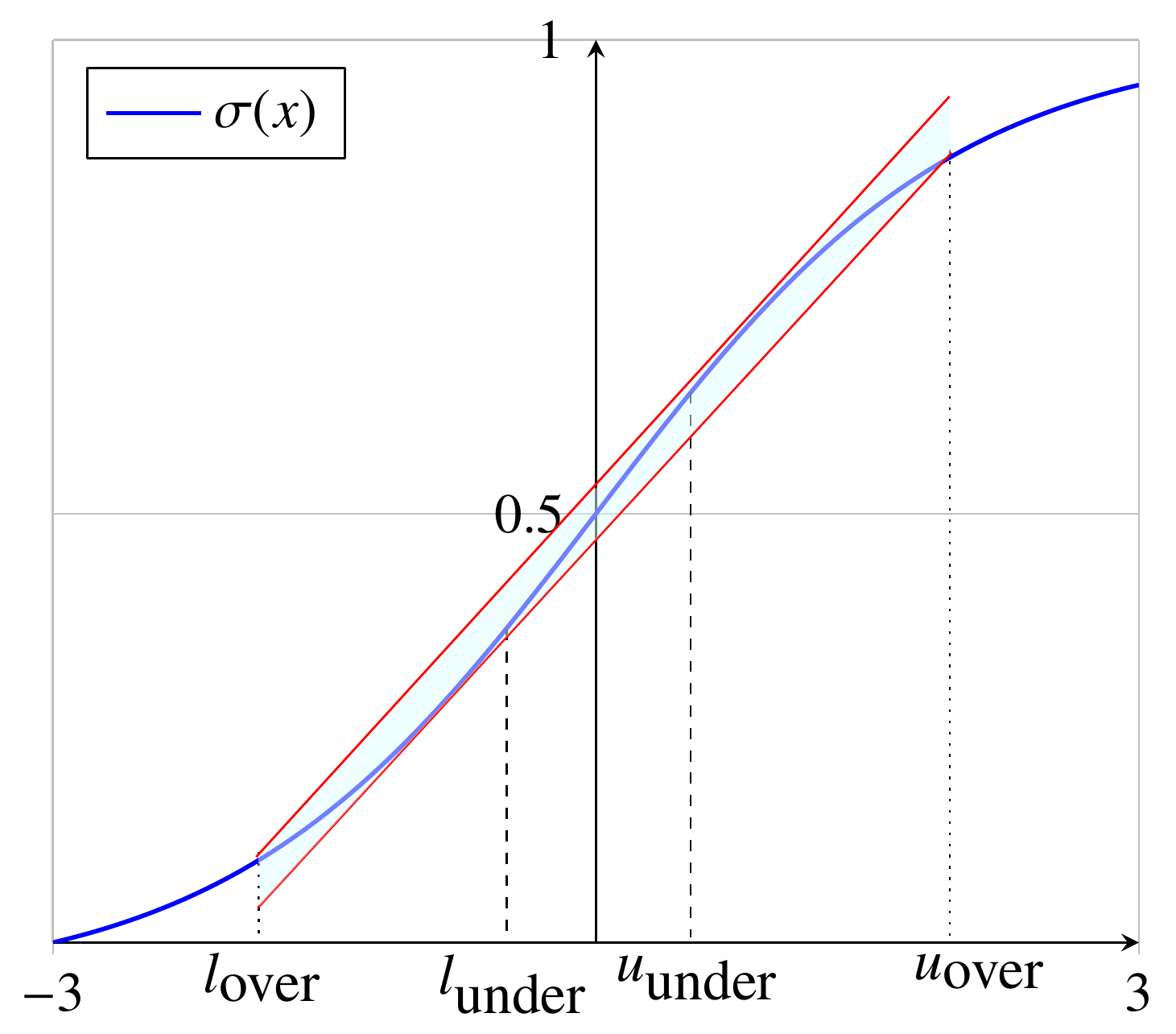}
		\caption{Case\,III: $\sigma'(l_U) \!<\! k\!\wedge\!\sigma'(u_U) \!<\! k$.}
		\label{fig:h}
	\end{subfigure}        
	\caption{The linear over-approximation based on overestimated and underestimated approximation domains.}
	\label{fig:figures}
	\vspace{-4mm}
\end{figure*}

\subsection{The Over-Approximation Strategy}

We omit the superscript and subscript and consider finding the approximation method of $\sigma(x)$ with the information of upper and lower approximation intervals. We assume that the lower approximation interval of input $x$ is $[l_L, u_L]$ and the upper approximation interval is $[l_U, u_U]$. As in \cite{wu2021tightening}, we consider three cases according to the relation between the slopes of $\sigma$ at the two endpoints of upper approximation interval $\sigma'(l_U)$, $\sigma'(u_U)$ and $k = \frac{\sigma(u_U)-\sigma(l_U)}{u_U-l_U}$.

\vspace{1ex}
\noindent\textbf{Case I.} When $\sigma'(l_U) < k < \sigma'(u_U)$, the line connecting the two endpoints is the upper bound. For the lower bound, the tangent line of $\sigma$ at $l_L$ is chosen if it is sound (Figure  \ref{fig:a}), otherwise the tangent line of $\sigma$ at $d$ crossing $(u_U, \sigma(u_U))$ is chosen (Figure \ref{fig:b}). Namely, we have $h_U(x) = k (x-u_U) + \sigma(u_U)$, and

\begin{equation} 
	h_L(x) = \left\{
	\begin{split}
		& \sigma'(l_L)(x - l_L) + \sigma(l_L), & l_L < d\\
		& \sigma'(d)(x - d) + \sigma(d), & l_L \ge d.
	\end{split}
	\right.
\end{equation}

\vspace{1ex}
\noindent\textbf{Case II.} When $\sigma'(l_U) > k > \sigma'(u_U)$, it is the symmetry of Case 1. the line connecting the two endpoints can be the lower bound. For upper bound, the tangent line of $\sigma$ at $u_L$ is chosen if it is sound (Figure \ref{fig:c}), otherwise the tangent line of $\sigma$ at $d$ crossing $(l_L, \sigma(l_L))$ is chosen (Figure  \ref{fig:d}). That is, $h_L(x) = k (x-l_U) + \sigma(l_U)$, and

\begin{equation} 
	h_U(x) = \left\{
	\begin{split}
		& \sigma'(u_L)(x - u_L) + \sigma(u_L), & u_L > d\\
		& \sigma'(d)(x - d) + \sigma(d), & u_L \le d.
	\end{split}
	\right.
\end{equation}

\vspace{1ex}
\noindent\textbf{Case III.} When $\sigma'(l_U) < k$ and $\sigma'(u_U) < k$, we first consider the upper bound. If the tangent line of $\sigma$ at $u_L$ is sound, we choose it to be the upper bound (Figure \ref{fig:e} and Fig \ref{fig:g}); otherwise we choose the tangent line of $\sigma$ at $d_1$ crossing $(l_L, \sigma(l_L))$ (Figure \ref{fig:f} and Figure \ref{fig:h}). Then we consider the lower bound. The tangent line of $\sigma$ at $l_L$ is chosen if it is sound (Figure \ref{fig:f} and Figure \ref{fig:g}), otherwise we choose the tangent line of $\sigma$ at $d_2$ crossing $(u_U, \sigma(u_U))$ (Figure \ref{fig:e} and Figure \ref{fig:h}). Namely, we have:

\begin{equation} 
	h_U(x) = \left\{
	\begin{split}
		& \sigma'(u_L)(x - u_L) + \sigma(u_L), & u_L > d_1\\
		& \sigma'(d_1)(x - d_1) + \sigma(d_1), & u_L \le d_1,
	\end{split}
	\right.
\end{equation}

\begin{equation} 
	h_L(x) = \left\{
	\begin{split}
		& \sigma'(l_L)(x - l_L) + \sigma(l_L), & l_L < d_2\\
		& \sigma'(d_2)(x - d_2) + \sigma(d_2), & l_L \ge d_2.
	\end{split}
	\right.
\end{equation}

The main idea of the approximation strategy is to make the overestimated input range as close as possible to the real one of each hidden neuron. As described in Theorem \ref{real_domain_tightness}, a preciser range allows us to define a tighter linear approximation. Under the premise of guaranteeing soundness, we use the guiding significance of the lower approximation interval to make the linear approximation closer to the real value of each node so as to obtain a more accurate approximation interval. Through layer-by-layer transmission, a more accurate output interval can be obtained, which thereby is used to compute more precise robustness verification results.

\begin{figure}
	\begin{center}
		\includegraphics[width=0.48\textwidth]{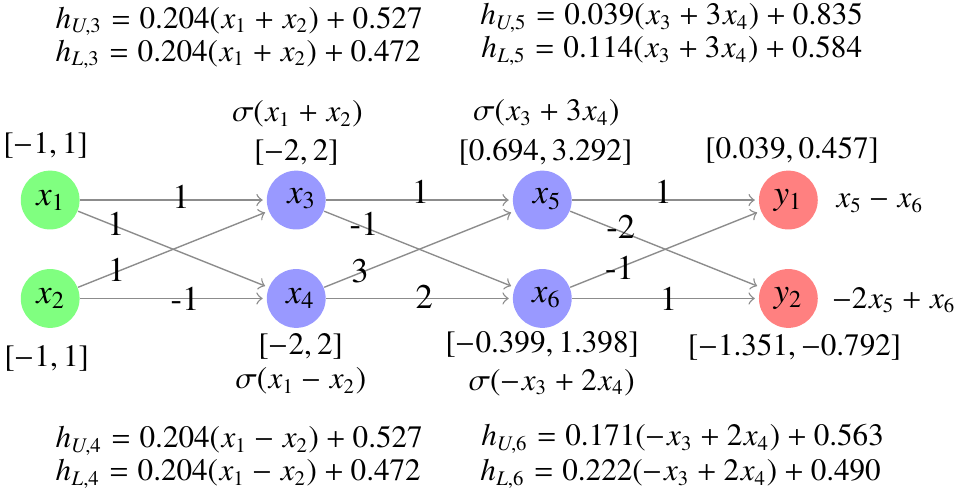}
		\caption{The approximations defined in our approach and the propagated intervals for the same network in Figure \ref{approximation_eg}.}
		\label{approx_eg_tightest}
	\end{center}
	\vspace{-7mm}
\end{figure}

\begin{example}\label{exm:3}
We reconsider the network in Figure \ref{approximation_eg} and define tighter over-approximations using our approach. Figure \ref{approx_eg_tightest} shows the approximations and the propagated intervals for the neurons on the hidden layers and the output layer. Because $x_3,x_4$ have precise input intervals, only $x_5,x_6$ need to under-approximate. Thus, we only need to redefine  their approximations according to our approach. We underestimate the input domains of $x_5,x_6$ and use them to guide the over-approximations in our approach. We achieve $9.74\%$ and $0.27\%$ reductions for the overestimations of $y_1,y_2$'s output ranges.  

\end{example}

Example \ref{exm:3} demonstrates the effectiveness of our proposed approximation approach even when it is applied to only one layer. The key  step of our approach from others is to compute corresponding underestimated input domains for the activation functions and guide theirs over-approximations.


\section{Under-Approximation Approaches}\label{sec:under}
In this section, we introduce two approaches, i.e., \emph{gradient-based} and \emph{sampling-based}, for underestimating the input domains for the activation functions on hidden neurons. The two approaches are complementary in that  the 
former is more efficient but computes less precise underestimated input domains while the latter performs in the opposite direction.

\subsection{The Sampling-Based Algorithm}

A simple yet efficient under-approximation approach is to randomly generate a number of valid samples and feed them into the network to track the reachable bounds of each hidden neuron's input. A sample is valid if the distance between it and the original input is less than a preset perturbation distance $\epsilon$. 


Algorithm \ref{sample_algorithm} shows the pseudo code of the sampling-based approach. First, we randomly generate 
$n$ valid samples from $\mathbb{B}_\infty(x_0, \epsilon)$ (Line 1) and initialize the lower and upper bounds $l_{L,r}^{(i)}$ and $u_{L,r}^{(i)}$ of each hidden neuron (Line 2). 
Then we feed each sample into the network (Line 3), record the input value  $v^{(i)}_{L,r}$ of each activation function (Line 6), and update the corresponding lower or upper bound by $v^{(i)}_{L,r}$ (Lines 7-8). 
The complexity of time of this algorithms is $O(n \sum_{i=1}^{m}k_i k_{i-1})$.


\setlength{\textfloatsep}{5pt}	
\begin{algorithm}[t]
	\SetKwData{Left}{left}\SetKwData{This}{this}\SetKwData{Up}{up}
	\SetKwFunction{Union}{Union}\SetKwFunction{FindCompress}{FindCompress}
	\SetKwInOut{Input}{Input}\SetKwInOut{Output}{Output}
	\caption{Sample-Based Under-Approximation.}
	\label{sample_algorithm}
	\Input{$F$: a network; $x_0$: an input to $F$; $\epsilon$: a $\ell_\infty$-norm radius; $n$: number of samples}
	\Output{$l_{L,r}^{(i)}, u_{L,r}^{(i)}$ for each hidden neuron $r$ on layer $i$}
	Randomly generate $n$ samples $S_n$ from $\mathbb{B}_\infty(x_0, \epsilon)$;\\
	$l_{L}\leftarrow -\infty, u_{L}\leftarrow \infty$\tcp*{Initialize all  upper and lower bounds}
	\For{each sample $x_p$ in $S_n$}{
		\For{each hidden layer $i$}{
			\For{each neuron $r$ on layer $i$}{
				$v_{p,r}^{(i)} := F_r^{(i)}(x_p)$\tcp*{Compute the output of  neuron $r$}
				$l_{L,r}^{(i)}\leftarrow \min (l_{L,r}^{(i)}, v_{p,r}^{(i)})$ \tcp*{Update $r$'s lower bound}
				$u_{L,r}^{(i)}\leftarrow \max (u_{L,r}^{(i)}, v_{p,r}^{(i)})$\tcp*{Update $r$'s upper  bound}
			}
		}
	}
	
\end{algorithm}

\subsection{The Gradient-Based Algorithm}\label{Gradient-Based_Approach}
\vspace{-1mm}

The conductivity of neural networks allows us to approximate the actual domain of each hidden neuron in Definition \ref{actual_domain_def} by gradient descent \cite{ruder2016overview}. The basic idea of gradient descent is to compute two valid samples according to the gradient of an  objective function to minimize and maximize the output value of the function, respectively. 
Using gradient descent, we can compute locally optimal lower and upper bounds as the underestimated input domains of activation functions.

Algorithm \ref{GD_algorithm} shows the pseudo code of the gradient-based approach. 
Its inputs include a neural network $F$, an input $x_0$ of $F$, an  $\ell_\infty$-norm radius $\epsilon$ and a step length $a$ of gradient descent. It returns an underestimated input domain for each neuron on the hidden layers. 
It first gets the function $F^{(i)}_r$ by Definition \ref{actual_domain_def} (Line 3), computes the gradient of $F^{(i)}_r$ and records its sign $\eta^{i}_r$ as the direction to update $x_0$ (Line 4). Then, 
the gradient descent is conducted with  one step forward  to generate a new input sample $x_{lower}$ (Line 5).  $x_{lower}$ is then modified to make sure it is in the normal ball. By feeding $x_{lower}$ to $F^{(i)}_r$, we obtain an under-approximated lower bound $l_{L,r}^{(i)}$ (Line 7). The upper bound can be computed likewise (Lines 8-10).

Considering the time complexity of Algorithm \ref{GD_algorithm}, we need to compute the gradient for each neuron on the $i$th hidden layer, of which time complexity is $O(\sum_{j=1}^i k_j k_{j-1})$. Thus, given an $m$-hidden-layer network, the time complexity of the gradient-based algorithm is $O(\sum_{i=1}^m k_i (\sum_{j=1}^i k_j k_{j-1}))$. 


\setlength{\textfloatsep}{5pt}	
\DecMargin{1em}
\begin{algorithm}[t]
	\SetKwData{Left}{left}\SetKwData{This}{this}\SetKwData{Up}{up}
	\SetKwFunction{Union}{Union}\SetKwFunction{FindCompress}{FindCompress}
	\SetKwInOut{Input}{Input}\SetKwInOut{Output}{Output}
	\caption{Gradient-Based Under-Approximation.}
	\label{GD_algorithm}
	\Input{$F$: a network; $x_0$: an input to $F$; $\epsilon$: a $\ell_\infty$-norm radius; $a$: the step length of gradient descent}
	\Output{$l_{L,r}^{(i)}, u_{L,r}^{(i)}$ for each  neuron $r$ in each hidden layer $i$}
	\For{each hidden layer $i={1,\ldots,m}$}{
		\For{each neuron $r$ on layer $i$}{
			Get the function $F_r^{(i)}$of neuron $r$\;
			$\eta_r^{(i)} \leftarrow sign(F_r^{(i)'}(x_0))$\tcp*{Get the sign of gradient of $r$}
			$x_{lower} \leftarrow x_0 - a\eta_r^{(i)} $\tcp*{One-step forward}
			Cut $x_{lower}$ s.t. $x_{lower} \in \mathbb{B}_\infty(x_0, \epsilon)$\tcp*{Make $x_{lower}$ valid}
			$l_{L,r}^{(i)}\leftarrow F_r^{(i)}(x_{lower})$\tcp*{Compute and store the lower bound}
			$x_{upper} \leftarrow x_0 + a\eta_r^{(i)} $\tcp*{Compute the upper case}
			Cut $x_{upper}$ s.t. $x_{upper} \in \mathbb{B}_\infty(x_0, \epsilon)$\\
			$u_{L,r}^{(i)}\leftarrow F_r^{(i)}(x_{upper})$ \tcp*{Compute and store the upper bound}
		}
	}
\end{algorithm}	


\section{Implementation and  Evaluation}\label{sec:exp}

For the evaluation purpose, we need to answer the following two research questions: 
\begin{enumerate}[(1)]
	\item \textit{How about the effectiveness and efficiency of our dual-approximation approach? }
	\item \textit{How to determine the hyper-parameters, i.e., the step length and the sampling number, in the gradient-based and sampling-based under-approximation algorithms?}
\end{enumerate}

\vspace{-2mm}
\subsection{Benchmarks and Experimental Settings}
\vspace{-1mm}
\noindent
\textit{Benchmarks.}
For comparison, we choose four state-of-the-art  approximation-based tools including NeWise \cite{2208.09872}, DeepCert \cite{wu2021tightening}, VeriNet \cite{HenriksenL20}, and RobustVerifier \cite{lin2019robustness}.  All are designed specifically for S-curved activation functions. 
We implemented our approach in a prototype tool \textsf{DualApp}. 
The tool is publicly available at \url{https://figshare.com/s/36048c6fb698e22bf13f}.   

\vspace{2mm}
\noindent
\textit{Datasets and Networks.} We collected and  trained totally 84 convolutional networks (CNNs) and fully-connected networks (FNNs) on image databases Mnist\cite{DBLP:journals/pieee/LeCunBBH98}, Fashion Mnist\cite{xiao2017/online} and Cifar-10\cite{krizhevsky2009learning}. The layers of  convolutional networks range from 4 to 10 for Mnist and Fashion Mnist. All of them have 5 filters in each layer. The number of neurons in the trained FNNs ranges from 60 to 1210 for Mnist and Fashion Mnist.  For Cifar-10, the trained CNNs consist of 3, 5, and 6 layers with different filters, respectively, while FNNs contain 510 and 2110 neurons. 
We trained three variant neural networks for some architecture using  Sigmoid, Tanh, and Arctan as activation functions, respectively.

\vspace{2mm}
\noindent
\textit{Experimental Settings.} 
We conducted all the experiments on a  workstation equipped with a 32-core AMD Ryzen Threadripper CPU @ 3.7GHz and 128GB RAM running Ubuntu 18.04.

\subsection{Experimental Results}

\begin{table*}[]
	\centering
	\def\arraystretch{0.9}
	\caption{Comparing the sampling-version \textsf{DualApp} (D.U.$_{sp}$) and four state-of-the-art tools including NeVise (N.W.), DeepCert (D.C.), VeriNet (V.N.) and RobustVerifier (R.V.) on the CNNs and FNNs with the Sigmoid activation function. ${\rm CNN}_{l-k}$ denotes a CNN with $l$ layers and  $k$ filters of size $3 \times 3$ on each layer. ${\rm FNN}_{l\times k}$ denotes a FNN with $l$ layers and $k$ neurons on each layer.}
	\label{tab:sigmoid}
	\setlength{\tabcolsep}{3pt}	
	\begin{tabular}{|l|r|r|r|r|r|r|r|r|r|r|r|r|r|}
		\hline
		\multirow{2}{*}{\textbf{Dataset}}                                                 & \multirow{2}{*}{\textbf{Model}} & \multirow{2}{*}{\textbf{Nodes}} & \textbf{D.U.$_{sp}$} & \multicolumn{2}{c|}{\textbf{N.W.}}              & \multicolumn{2}{c|}{\textbf{D.C.}}            & \multicolumn{2}{c|}{\textbf{V.N.}}             & \multicolumn{2}{c|}{\textbf{R.V.}}      & \multirow{2}{*}{\begin{tabular}[c]{@{}c@{}}\textbf{D.U.$_{sp}$} \\ \textbf{Time} \textbf{(s)}\end{tabular}} & \multirow{2}{*}{\begin{tabular}[c]{@{}c@{}}\textbf{Others Time(s)}\end{tabular}} \\ \cline{4-12}
		&                        &                        &  {Bounds}  & {Bounds}  & Impr. (\%) & {Bounds}  & Impr. (\%)  & {Bounds}  & Impr. (\%) & {Bounds}  & Impr. (\%) &                                                                               &                                 \\ \hline
		\multirow{7}{*}{Mnist}                                                    & ${\rm CNN}_{4-5}$      & 8,690                  & 0.05819  & {0.05698} & 2.12      & {0.05394} & 7.88      & {0.05425} & 7.26      & {0.05220} & 11.48     & 14.70                                                                         & 0.98   $\pm$ 0.02  \\
		& ${\rm CNN}_{5-5}$      & 10,690                 & 0.05985  & {0.05813} & 2.96      & {0.05481} & 9.20      & {0.05503} & 8.76      & {0.05125} & 16.78     & 20.13                                                                         & 2.67   $\pm$ 0.29  \\
		& ${\rm CNN}_{6-5}$      & 12,300                 & 0.06450   & {0.06235} & 3.45      & {0.05898} & 9.36      & {0.05882} & 9.66      & {0.05409} & 19.25     & 25.09                                                                         & 4.86   $\pm$ 0.34  \\
		& ${\rm CNN}_{8-5}$     & 14,570                 & 0.11412  & {0.09559} & 19.38     & {0.08782} & 29.95     & {0.08819} & 29.40     & {0.06853} & 66.53     & 34.39                                                                         & 11.89   $\pm$ 0.21 \\
		& ${\rm FNN}_{5\times 100}$                  & 510                    & 0.00633  & {0.00575} & 10.09     & {0.00607} & 4.28      & {0.00616} & 2.76      & {0.00519} & 21.97     & 7.10                                                                           & 0.79   $\pm$ 0.05  \\
		& ${\rm FNN}_{6\times 200}$                  & 1,210                  & 0.02969  & {0.02909} & 2.06      & {0.02511} & 18.24     & {0.02829} & 4.95      & {0.01811} & 63.94     & 8.64                                                                         & 2.82   $\pm$ 0.34  \\ \hline
		\multirow{7}{*}{\begin{tabular}[c]{@{}l@{}}Fashion \\ Mnist\end{tabular}} & ${\rm CNN}_{4-5}$      & 8,690                  & 0.07703  & {0.07473} & 3.08      & {0.07204} & 6.93      & {0.07200} & 6.99      & {0.06663} & 15.61     & 15.26                                                                         & 1.06   $\pm$ 0.09  \\
		& ${\rm CNN}_{5-5}$      & 10,690                 & 0.07288  & {0.07044} & 3.46      & {0.06764} & 7.75      & {0.06764} & 7.75      & {0.06046} & 20.54     & 20.95                                                                         & 3.18   $\pm$ 0.42  \\
		& ${\rm CNN}_{6-5}$      & 12,300                 & 0.07655  & {0.07350}  & 4.15      & {0.06949} & 10.16     & {0.06910} & 10.78     & {0.06265} & 22.19     & 25.96                                                                         & 5.63   $\pm$ 0.77  \\
		& ${\rm CNN}_{8-5}$      & 14,570                 & 0.14119  & {0.14551} & -2.97     & {0.12448} & 13.42     & {0.12376} & 14.08     & {0.08246} & 71.22     & 37.36                                                                          & 13.22   $\pm$ 0.98 \\
		& ${\rm FNN}_{1\times 50}$                   & 60                     & 0.03616  & {0.03284} & 10.11     & {0.03511} & 2.99      & {0.03560} & 1.57      & {0.02922} & 23.75     & 0.84                                                                          & 0.02   $\pm$ 0.00    \\
		& ${\rm FNN}_{5\times 100}$                  & 510                    & 0.00801  & {0.00710}  & 12.82     & {0.00776} & 3.22      & {0.00789} & 1.52      & {0.00656} & 22.10     & 2.98                                                                          & 0.65   $\pm$ 0.00     \\ \hline
		\multirow{5}{*}{Cifar-10}                                                 & ${\rm CNN}_{3-2}$    & 2,514                  & 0.03197  & {0.03138} & 1.88      & {0.03120} & 2.47      & {0.03119} & 2.50      & {0.03105} & 2.96      & 5.54                                                                          & 0.32   $\pm$ 0.02  \\
		& ${\rm CNN}_{5-5}$     & 10,690                 & 0.01973  & {0.01926} & 2.44      & {0.01921} & 2.71      & {0.01913} & 3.14      & {0.01864} & 5.85      & 31.45                                                                         & 4.86   $\pm$ 0.41  \\
		& ${\rm CNN}_{6-5}$      & 12,300                 & 0.02338  & {0.02289} & 2.14      & {0.02240} & 4.38      & {0.02234} & 4.66      & {0.02124} & 10.08     & 43.51                                                                         & 10.53   $\pm$ 0.67 \\
		& ${\rm FNN}_{5\times 100}$                 & 510                    & 0.00370  & {0.00329} & 12.46      & {0.00368} & 0.54      & {0.00368} & 0.54      & {0.00331} & 11.78      & 2.97                                                                          & 0.64   $\pm$ 0.01  \\
		& ${\rm FNN}_{3\times 700}$                    & 2,110                  & 0.00428  & {0.00348}  & 22.99      & {0.00427} & 0.23     & {0.00426} & 0.47      & {0.00397} & 7.81      & 32.68                                                                         & 10.85   $\pm$ 0.58 \\ \hline
	\end{tabular}
	\label{exp_only_sample}
	\vspace{-1mm}
\end{table*} 

\begin{table*}[]
	\centering
	\def\arraystretch{0.9}
	\caption{Comparison between sampling-based and existing tools on Tanh and Arctan networks.}
	\setlength{\tabcolsep}{3pt}	
	\begin{tabular}{|l|r|r|r|r|r|r|r|r|r|r|r|r|r|}
		\hline
		\multirow{2}{*}{\textbf{Dataset}} & \multirow{2}{*}{\textbf{Model}} & \multirow{2}{*}{\textbf{$\sigma$}} & \textbf{D.U.$_{sp}$} & \multicolumn{2}{c|}{\textbf{N.W.}}  & \multicolumn{2}{c|}{\textbf{D.C.}} & \multicolumn{2}{c|}{\textbf{V.N.}}             & \multicolumn{2}{c|}{\textbf{R.V.}}      & \multirow{2}{*}{\begin{tabular}[c]{@{}c@{}}\textbf{D.U.$_{sp}$} \\ \textbf{Time} \textbf{(s)}\end{tabular}} & \multirow{2}{*}{\begin{tabular}[c]{@{}c@{}}\textbf{Others Time(s)} \end{tabular}} \\ \cline{4-12}
		&                        &                        &  {Bounds}  & {Bounds}  & Impr. (\%) & {Bounds}  & Impr. (\%)  & {Bounds}  & Impr. (\%) & {Bounds}  & Impr. (\%) &                                                                               &                                 \\ \hline
		\multirow{6}{*}{Mnist}                                                   

		& \multirow{2}{*}{${\rm CNN}_{5-5}$}   & Tanh   & 0.01501  & 0.01224 & 22.63 & 0.01486  & 1.01  & 0.01481 & 1.35 & 0.01281 & 17.17 & 19.75 & 2.38 \textpm 0.32  \\
		&                                      & Arctan & 0.01503  & 0.01232 & 22.00 & 0.01286  & 16.87 & 0.01488 & 1.01 & 0.01253 & 19.95 & 14.17 & 2.52 \textpm 0.22  \\
		& \multirow{2}{*}{${\rm CNN}_{6-5}$}   & Tanh   & 0.01244  & 0.00997 & 24.77 & 0.01200    & 3.67  & 0.01204 & 3.32 & 0.00976 & 27.46 & 24.87 & 4.80 \textpm 0.57  \\
		&                                      & Arctan & 0.01610  & 0.01288 & 25.00 & 0.01213  & 32.73 & 0.01565 & 2.88 & 0.01285 & 25.29 & 18.61 & 5.11 \textpm 0.72  \\
		& \multirow{2}{*}{${\rm FNN}_{5x100}$} & Tanh   & 0.00476  & 0.00343 & 38.78 & 0.00464  & 2.59  & 0.00469 & 1.49 & 0.00389 & 22.37 & 6.74  & 0.88 \textpm 0.10  \\
		&                                      & Arctan & 0.00498  & 0.00367 & 35.69 & 0.00354  & 40.68 & 0.00489 & 1.84 & 0.00405 & 22.96 & 6.45  & 0.74 \textpm 0.15  \\
		\hline
		\multirow{4}{*}{\begin{tabular}[c]{@{}l@{}}Fashion \\ Mnist\end{tabular}} 
		& \multirow{2}{*}{${\rm CNN}_{6-5}$}   & Tanh   & 0.01239  & 0.00993 & 24.77 & 0.01210   & 2.40  & 0.01210  & 2.40 & 0.01042 & 18.91 & 21.01 & 5.28 \textpm 0.82  \\
		&                                      & Arctan & 0.01606  & 0.01260  & 27.46 & 0.01299  & 23.63 & 0.01571 & 2.23 & 0.01300  & 23.54 & 18.38 & 4.75 \textpm 0.27  \\
		& \multirow{2}{*}{${\rm FNN}_{5x100}$} & Tanh   & 0.00320  & 0.00227 & 40.97 & 0.00312  & 2.56  & 0.00314 & 1.91 & 0.00259 & 23.55 & 6.91  & 0.75 \textpm 0.05  \\
		&                                      & Arctan & 0.00332  & 0.00240  & 38.33 & 0.00235  & 41.28 & 0.00327 & 1.53 & 0.00270  & 22.96 & 7.38  & 1.06 \textpm 0.08  \\
		\hline
		\multirow{6}{*}{Cifar-10}                                                 
		& \multirow{2}{*}{${\rm CNN}_{5-5}$}   & Tanh   & 0.00932  & 0.00762 & 22.31 & 0.00932  & 0.00  & 0.00930  & 0.22 & 0.00882 & 5.67  & 9.88  & 0.51   \textpm 0.07  \\
		&                                      & Arctan & 0.00960   & 0.00792 & 21.21 & 0.00901  & 6.55  & 0.00958 & 0.21 & 0.00914 & 5.03  & 9.03  & 0.51   \textpm 0.05  \\
		& \multirow{2}{*}{${\rm CNN}_{6-5}$}   & Tanh   & 0.00655  & 0.00524 & 25.00 & 0.00655  & 0.00  & 0.00653 & 0.31 & 0.00621 & 5.48  & 18.01 & 1.98   \textpm 0.65  \\
		&                                      & Arctan & 0.00780   & 0.00624 & 25.00 & 0.00656  & 18.90 & 0.00777 & 0.39 & 0.00738 & 5.69  & 17.85 & 2.16   \textpm 0.42  \\
		& \multirow{2}{*}{${\rm FNN}_{5x100}$} & Tanh   & 0.00187  & 0.00131 & 42.75 & 0.00185  & 1.08  & 0.00185 & 1.08 & 0.00163 & 14.72 & 23.58 & 12.29  \textpm 2.74  \\
		&                                      & Arctan & 0.00188  & 0.00136 & 38.24 & 0.00140   & 34.29 & 0.00187 & 0.53 & 0.00168 & 11.90 & 176.22 & 115.06 \textpm 14.03  \\
		\hline
	\end{tabular}
\vspace{-3mm}
	\label{exp_only_sample_tan_atan}
\end{table*} 


\noindent 
\textit{Experiment I.} Table \ref{exp_only_sample} shows the comparison results between sampling-based approximation and other four benchmarks on 17 networks with Sigmoid activation function. We randomly chose 100 inputs from each dataset and computed the average of their certified lower bounds. For each input image, we took 1000 samples for guiding approximation. 
The results show that our approach outperforms all four competitors almost in all cases. 
 The improvement can be up to 71.22\% when compared with RobustVerifier. 
 The only one exception is the network CNN$_{\text{8-5}}$ against NeWise. That is because NeWise is provably the tightest when the network is monotonous. We also found that, in general, the deeper the network is, the more the certified lower bound obtained by our approach is improved. 
Table \ref{exp_only_sample_tan_atan} partially presents the results on the Tanh and Arctan networks. In these models, our approach achieves up to 42.75\% improvement. The complete results can be found in the supplementary appendix. 

Regarding efficiency, our sampling-based approach takes a little more time than other tools because of the sampling procedure. However, these extra time-consuming can be ignored compared with the improvement of verification results. We use $t \pm e$ to denote the other four tools' average time cost $t$ and the size of the interval $2e$.   

We also evaluated the performances of the sampling-based and gradient descent-based algorithms in the verification. 
We considered 8 FNNs and 8 CNNs with the Sigmoid activation function in our experiment. 
For diversity, half of the networks are trained on Mnist and the other half on Fashion Mnist.

The comparison results are shown in Figure \ref{GDSampling}.
The abscissa denotes different network architecture,
the primary axis represents the value of certified-lower bound, and the secondary axis represents the time consuming in seconds.
It can be observed that the certified lower bounds computed by the gradient descent-based algorithm are always larger than the ones obtained by the sampling-based one. 
That is because the gradient descent-based approximation approach can theoretically compute tighter under-approximation bounds. 
This also reflects the effect of precise under-approximated bounds on defining tight over-approximations. 


Note that Figure \ref{GDSampling} shows that the gradient descent-based version costs more time than the sampling-based one, which is consistent with the time complexity of the two algorithms. If there is no time constraint, we can choose the gradient-based version for a more precise verification result. Otherwise, the sampling-based version is more suited to balancing the precision and time cost. 

\begin{figure*}
	\centering
	\begin{subfigure}{0.25\textwidth}
		\includegraphics[width=0.9\textwidth]{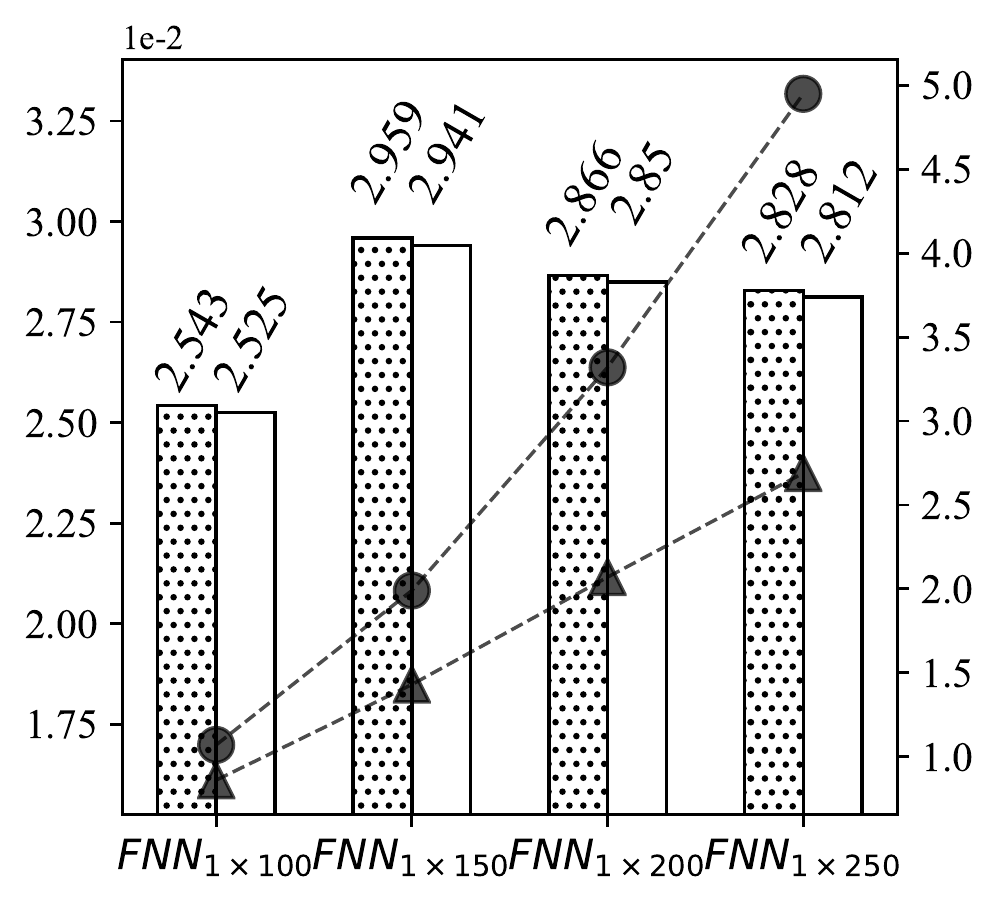}
		\caption{FNNs on Mnist.}
		\label{fig:GDSampling_mnist}
	\end{subfigure}
	\hfill
	\begin{subfigure}{0.24\textwidth}
		\includegraphics[width=0.9\textwidth]{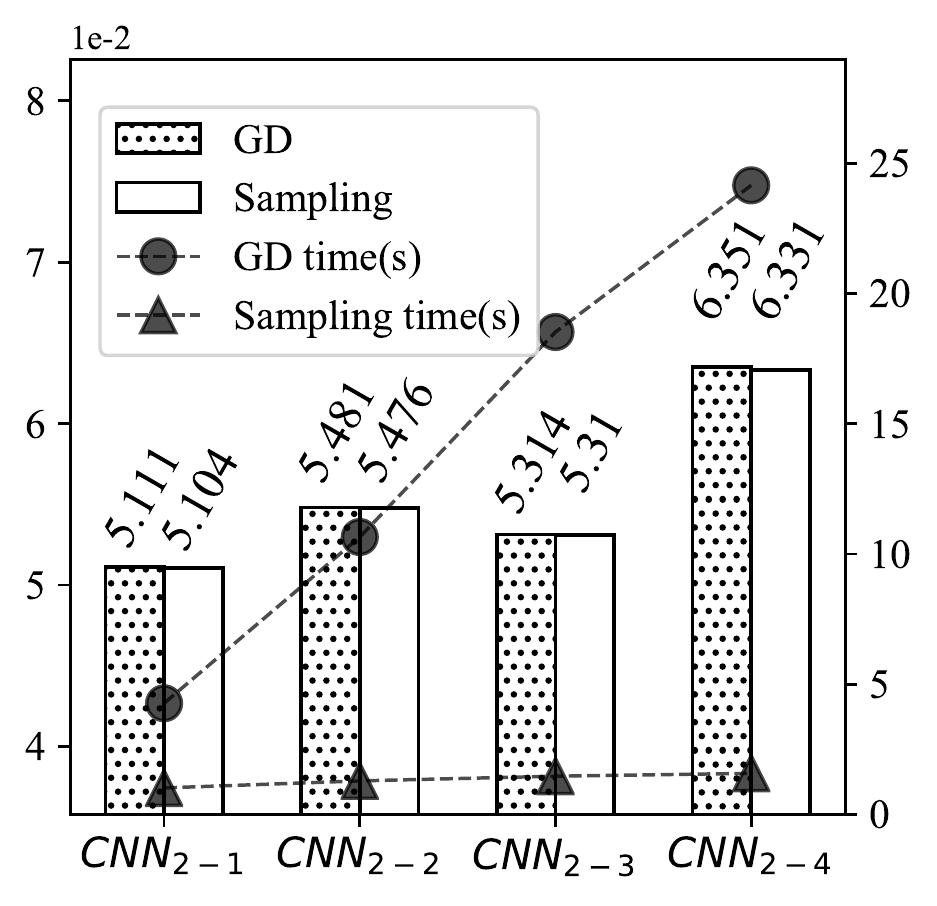}
		\caption{CNNs on  Mnist.}
		\label{fig:GDSampling_fashion_mnist}
	\end{subfigure}
	\hfill
	\centering
	\begin{subfigure}{0.24\textwidth}
		\includegraphics[width=0.9\textwidth]{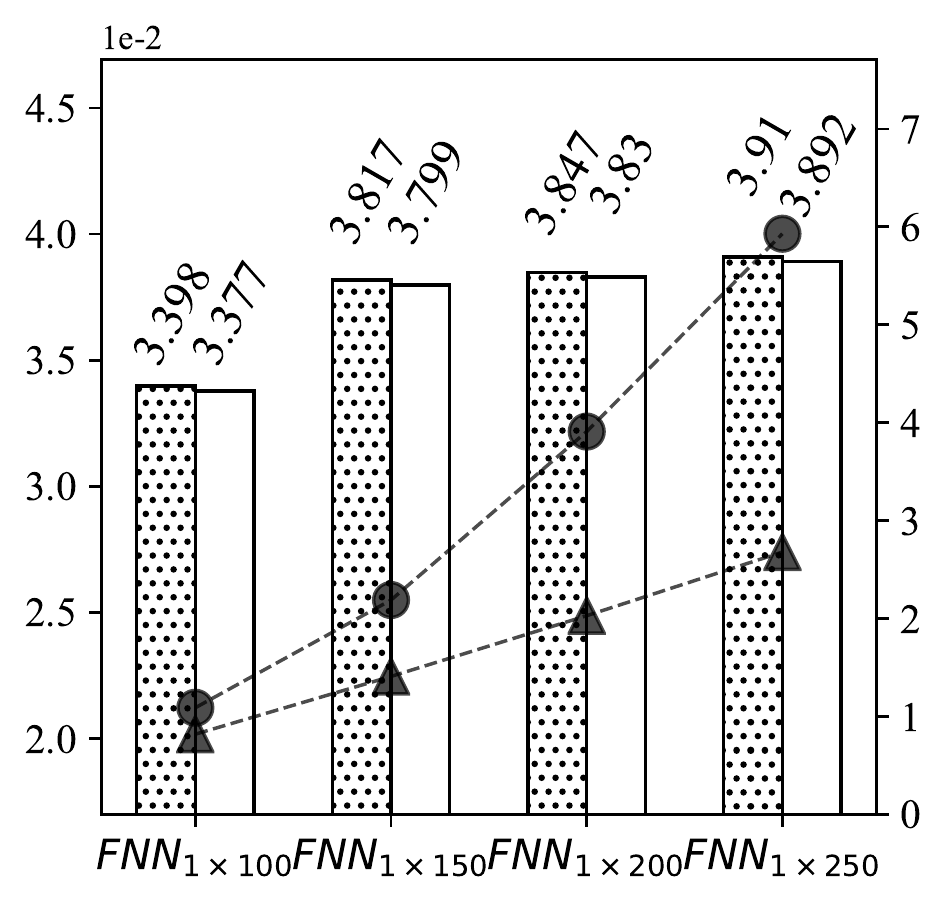}
		\caption{FNNs on Fashion Mnist.}
		\label{fig:GDSampling_mnist}
	\end{subfigure}
	\hfill
	\begin{subfigure}{0.24\textwidth}
		\includegraphics[width=0.9\textwidth]{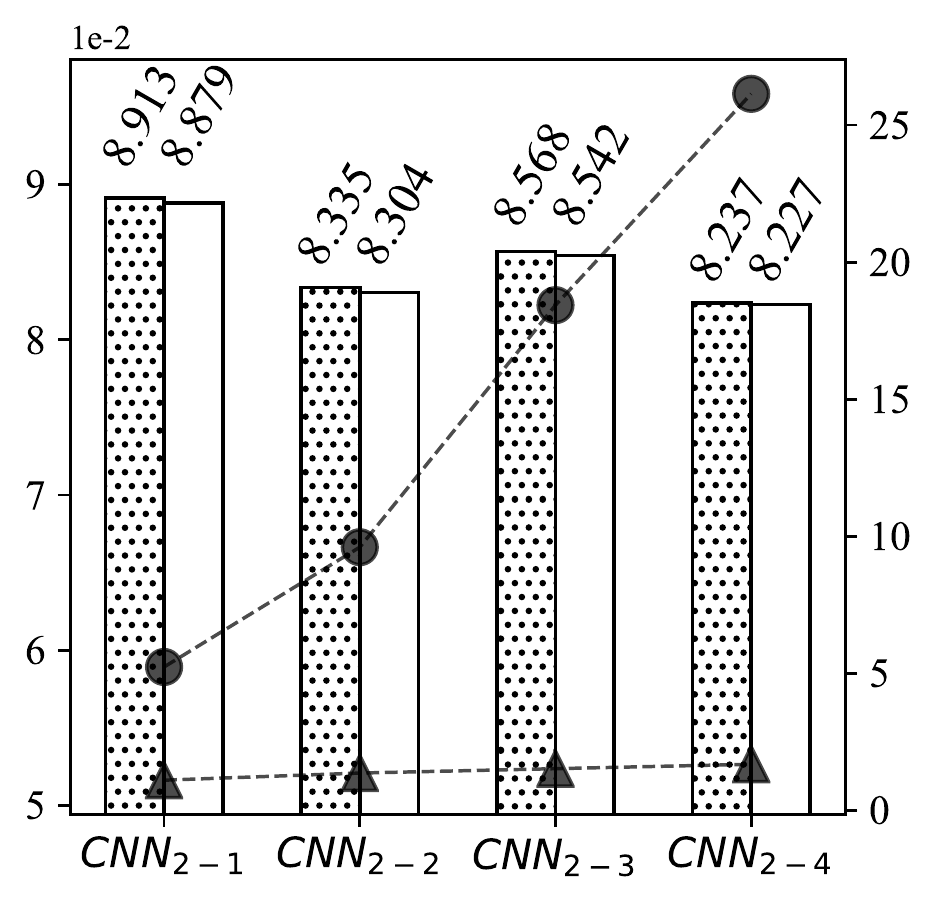}
		\caption{CNNs on Fashion Mnist.}
		\label{fig:GDSampling_fashion_mnist}
	\end{subfigure}
 \vspace{-1mm}
	\caption{Comparison between sampling-based and gradient-based algorithms for the robustness verification of 16 neural networks.}
 	\vspace{-5mm}
     \label{GDSampling}
\end{figure*}

\begin{figure}
	\centering
	\begin{subfigure}{0.23\textwidth}
		\includegraphics[width=\textwidth]{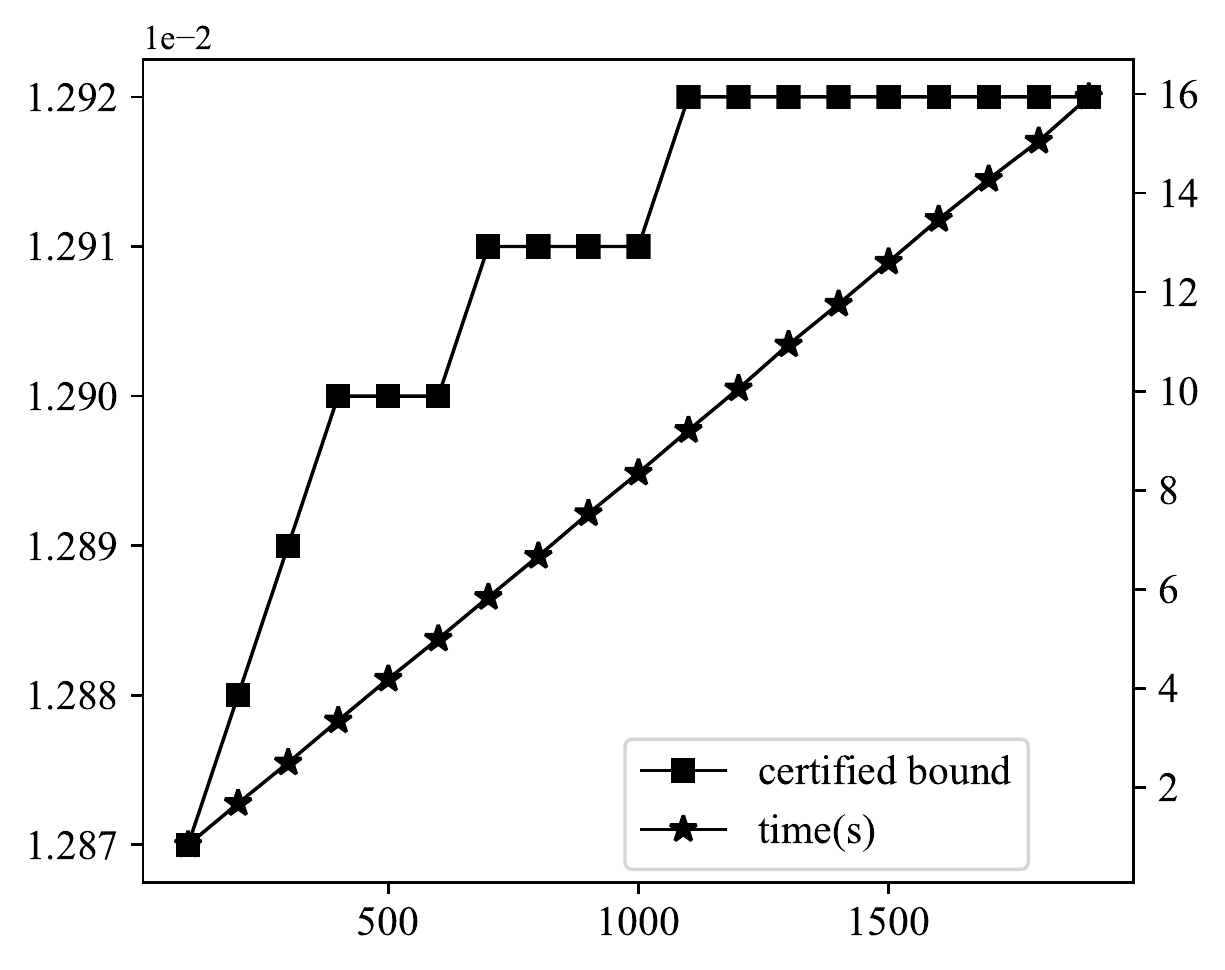}
		\caption{Sampling-based}
		\label{fig9:Sampling_fashion_mnist_fnn}
	\end{subfigure}
	\hfill
	\begin{subfigure}{0.23\textwidth}
		\includegraphics[width=\textwidth]{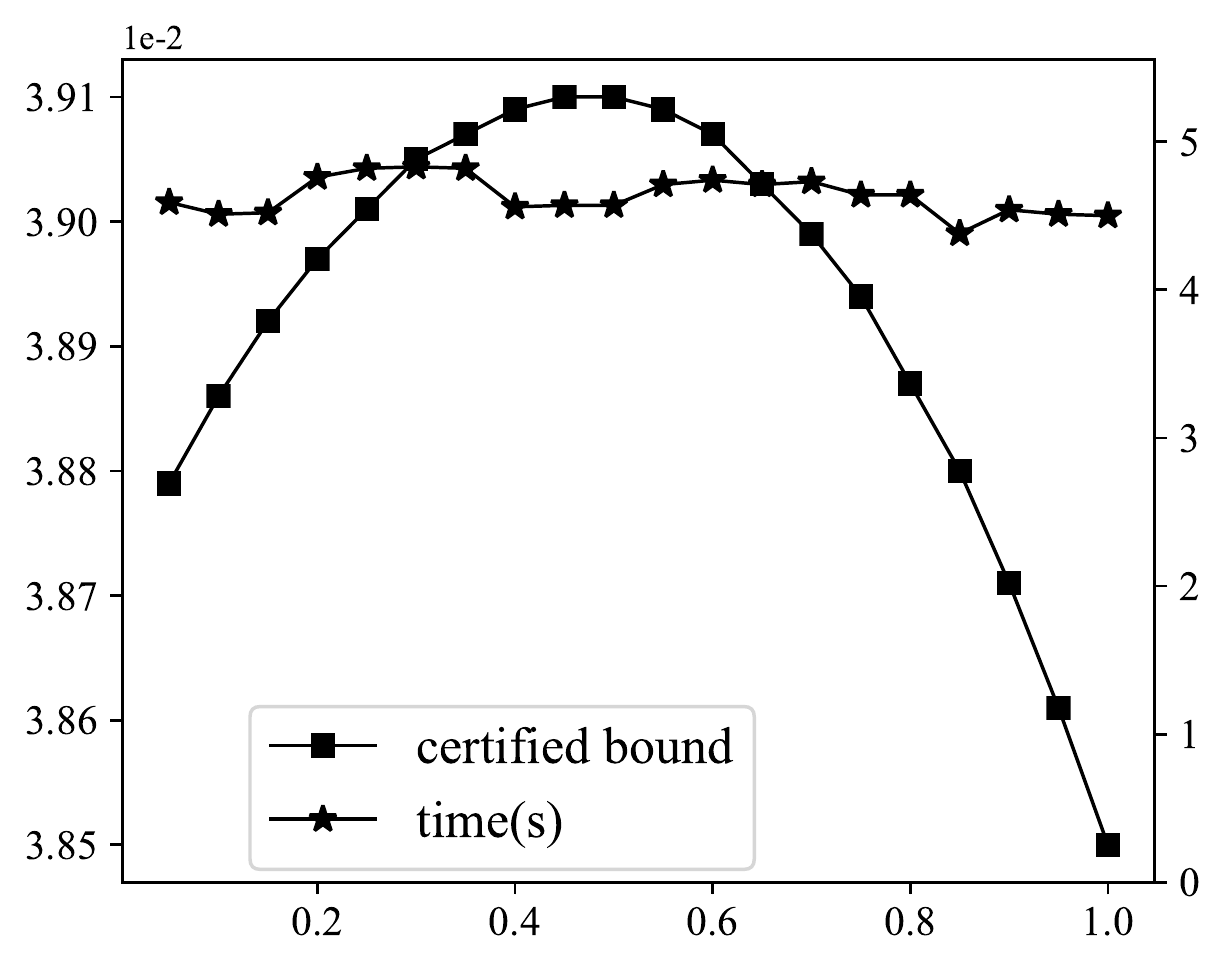}
		\caption{Gradient-based}
		\label{fig9:GD_fashion_mnist_fnn}
	\end{subfigure}
	\hfill
	\caption{The effects of hyper-parameters in the  sampling-based and gradient-based  under-approximation algorithms.}
	\label{hyper_para_fig9}
\end{figure}



\noindent 
\textit{Experiment II.}  To answer research question 2, we 
conduct experiments on 4 neural networks
 trained on Mnist and Fashion-Mnist, respectively. 
 Figure \ref{fig9:Sampling_fashion_mnist_fnn} shows the relations between the computed certified lower bound, the number of samples, and the time cost for an FNN trained on Mnist. 
 The computed bound is monotonously increasing with the number of samples. The increase becomes slow and almost stops when the number of samples reaches over 1000. That is because the underestimated approximation domain cannot be improved with more samples. As for efficiency, there is a linear relation between the number of samples and the time cost. The results for the other three networks are similar, and we put them in the appendix due to the space limit. 
 


We conduct similar experiments on the gradient descent-based algorithm. We consider the step length from $0.05\epsilon$ to $\epsilon$ by step of $0.05\epsilon$. Figure \ref{fig9:GD_fashion_mnist_fnn} shows the relation between the computed bound and the step length for an FNN trained on Mnist. We observe that when the step length is set around $0.45\epsilon$, 
the computed bound is maximal. This finding is also applicable to other networks. The time cost is almost the same and independent of the step length, as shown in Figure \ref{fig9:GD_fashion_mnist_fnn}.

%

\vspace{-1mm}
\subsection{Threats to Validity}
\vspace{-1mm}

It is worth mentioning that 
our approach may not always perform better than NeWise, as shown in Table \ref{tab:sigmoid}. That is because NeWise defines the provably tightest approximations when neural networks are monotonous \cite{2208.09872}. A sufficient condition for the monotonousness is when all the weights in the networks are non-negative. This special type of networks are often used in some specific applications to autoencoding \cite{ali2017automatic,DBLP:conf/icassp/NeacsuPB20} and  the detection of malware \cite{ceschin2019shallow,kargarnovin2021mal2gcn} and spam \cite{DBLP:journals/corr/abs-1806-06108}. 
A network may be  monotonous even if it contains both positive and negative weights. In such  case, NeWise is more precise in the verification than \textsf{DualApp}. However, our approach is complementary to NeWise and is applicable to more general non-monotonous networks. 

Another possible threat is the determining of an appropriate number of sampled inputs and a step length for the gradient-based approach. Our experimental results statistically show that 1000 samples and 0.45 step length are enough to achieve tight approximations. 
However, there may be better choices for these hyper-parameters for defining tighter approximations. 

	\section{Related Work}\label{sec:rel}
This work is a sequel to several works on neural network robustness verification based on approximations. We classify them into two categories.

\noindent \textit{Over-approximation approaches.}
Due to the intrinsic complexity in the neural network robustness verification, approximating the non-linear activation functions is the mainstreaming approach for scalability. Zhang \emph{et al.} defined three cases for over-approximating S-curved activation functions \cite{zhang2018efficient}. Wu and Zhang proposed a fine-grained approach and identified five cases for defining tighter approximations \cite{wu2021tightening}. Lyu \emph{et al.} proposed to define tight approximations by optimization at the price of sacrificing efficiency. Henriksen and Lomuscio \cite{HenriksenL20} defined tight approximations by minimizing the gap area between the bound and the curve. However, all these approaches are proved superior to others only on specific networks \cite{2208.09872}. The approximation approach proposed in the work \cite{2208.09872} is proved to be the tightest when the networks are monotonous. All these approaches only consider overestimated approximation domains. Paulsen and Wang recently proposed an interesting approach for synthesizing tight approximations guided by generated examples \cite{paulsen2022example,paulsen2022linsyn}. Their approach shares a similar idea to ours by computing sound and tight over-approximations from unsound templates. However, their approach needs global optimization techniques to guarantee soundness, while our approach ensures the soundness of individual neurons statistically. 



\noindent \textit{Under-approximation approaches.} The essence of under-approximation in our approach is to estimate the lower and upper input bounds of activation functions. There are several related approaches based on white-box attacks \cite{chakraborty2018adversarial} and testings \cite{he2020towards}. For instance, 
the fast gradient sign method (FGSM) \cite{goodfellow2014explaining} is a well-known approach for generating adversarial examples to intrigue corner cases for classifications. Other attack approaches include C\&W \cite{carlini2017towards}, DeepFool \cite{moosavi2016deepfool}, JSMA \cite{papernot2016limitations}, etc. 
The white-box testing for neural networks is to generate specific test cases to intrigue target neurons under different coverage criteria. Various adversarial sample generation approaches have been proposed \cite{lee2020effective,sun2019structural,yu2022white,guo2018dlfuzz}. We believe these attack and testing approaches can be tailored for under-approximations.


\section{Concluding Remarks and Future Work}\label{sec:conc}

We have proposed an under-approximation guided approach to defining tight over-approximations for the robustness verification of deep neural networks. We identified another important factor, i.e., \emph{approximation domain}, which is missed by almost all the existing approximation approaches for defining tight over-approximations. We analyzed over-approximations that have dual overestimation effects and demonstrated that under-approximation can effectively reduce the overestimation. We proposed two complementary under-approximation approaches and  implemented a prototype tool \textsf{DualApp} to evaluate our approaches extensively on a suite of benchmarks. The experimental results demonstrated that \textsf{DualApp} outperformed the state-of-the-art tools with up to 71.22\% improvement to certified robustness bounds. 

Our dual-approximation approach can be integrated into other abstraction-based neural network verification approaches \cite{gehr2018ai2,singh2019abstract,zhang2020detecting}. That is because all these approaches require non-linear activation functions that shall be over-approximated to handle abstract domains.  Besides robustness verification, we believe our approach is also applicable to the variants of robustness verification problems,  such as fairness \cite{bastani2019probabilistic} and $\epsilon$-weekend robustness \cite{huang2022}. The verification of these properties can be reduced to optimization problems containing the nonlinear activation functions in networks. 

The dual-approximation approach also sheds light on a new promising direction of combing testing approaches and verification techniques. For neural network testing, there are several under-approximation approaches, \textit{e.g.} fuzzing and mutating samples to activate target neurons in a neural network. These approaches can be integrated into our approach to compute less underestimated approximation domains more efficiently. We would consider the integration as our future work.

\onecolumn \begin{multicols}{2}
\bibliographystyle{IEEEtran}
\bibliography{icse23}	
\end{multicols}
%
%

\end{document}